\documentclass[preprint,12pt,3p]{elsarticle}

\usepackage{booktabs}
\usepackage{multirow}   
\usepackage{url}
\usepackage{float}
\usepackage{lipsum}
\makeatletter
\def\ps@pprintTitle{%
 \let\@oddhead\@empty
 \let\@evenhead\@empty
 \def\@oddfoot{}%
 \let\@evenfoot\@oddfoot}
\makeatother
\usepackage{caption}
\usepackage{algorithm}
\usepackage[noend]{algpseudocode}
\usepackage{subcaption}
\usepackage{hyperref} 
 \bibpunct[, ]{(}{)}{,}{a}{}{,}%
 \def\bibsep{\smallskipamount}%
\hypersetup{colorlinks=true}
\usepackage{verbatim}
 \bibpunct[, ]{(}{)}{,}{a}{}{,}%
 \def\bibsep{\smallskipamount}%
\usepackage{amssymb}
\usepackage{nomencl}
\makenomenclature

\usepackage[utf8]{inputenc}
\usepackage[english]{babel}

\usepackage{blindtext}
\usepackage{CJK}
\usepackage{amsmath}
\usepackage{setspace}
\usepackage{caption}
\usepackage{graphicx}
\usepackage{tabularx}
\usepackage{mwe}
\usepackage{color}
\usepackage{comment}

\usepackage{amsthm}
\newtheorem{myDef}{Definition}
\newtheorem{myCla}{Claim}

\newtheorem{myProp}{Proposition}
\newtheorem{myAsum}{Assumption}

\newtheorem{myHeu}{Heuristic}
\newtheorem{myPro}{Property}
\newcommand\blfootnote[1]{%
  \begingroup
  \renewcommand\thefootnote{}\footnote{#1}%
  \addtocounter{footnote}{-1}%
  \endgroup
}

\newtheorem{myEx}{Example}

\journal{Transportation Research Part B}

\begin{document}
\singlespacing
\begin{frontmatter}

\title{Stochastic on-time arrival problem in transit networks}
\blfootnote{Accepted manuscript for Transportation Research Part B, https://doi.org/10.1016/j.trb.2018.11.013}
\blfootnote{\textcopyright2018. This manuscript version is made available under the CC-BY-NC-ND 4.0 license http://creativecommons.org/licenses/by-nc-nd/4.0/}

\author[label1]{Yang Liu\corref{cor1}}
\address[label1]{School of Civil and Environmental
Engineering, Cornell University, Ithaca, NY 14853}

\cortext[cor1]{Corresponding author}
\ead{yl2464@cornell.edu}

\author[label2]{Sebastien Blandin}
\address[label2]{IBM Research, Singapore \fnref{label4}}
\ead{sblandin@sg.ibm.com}

\author[label1]{Samitha Samaranayake}
\ead{samitha@cornell.edu}





\begin{abstract}
This article considers the stochastic on-time arrival problem in transit networks where both the travel time and the waiting time for transit services are stochastic. A specific challenge of this problem is the combinatorial solution space due to the unknown ordering of transit line arrivals. We propose a network structure appropriate to the online decision-making of a passenger, including boarding, waiting and transferring. In this framework, we design a dynamic programming algorithm that is pseudo-polynomial in the number of transit stations and travel time budget, and exponential in the number of transit lines at a station, which is a small number in practice. To reduce the search space, we propose a definition of transit line dominance, and techniques to identify dominance, which decrease the computation time by up to $90\%$ in numerical experiments. Extensive numerical experiments are conducted on both a synthetic network and the Chicago transit network. 
\end{abstract}


\end{frontmatter}


\section{Introduction}\label{sec:intro}
People navigate over 1 billion kilometers a day using mobile routing services~\citep{recode}, and many of these services provide travelers with information on transit networks. In most of such applications, given an origin-destination (OD) pair and a desired departure or arrival time, the route associated with the minimum expected trip time is provided to the user. However, the uncertainty associated with these recommendations, either due to variability in travel time or the transit service headway, is rarely accounted for. In contrast, a number of surveys and numerical studies have highlighted that the degree of risk aversion of transit passengers highly affects their route choices~\citep{szeto2011reliability}, even more so than in the road networks. In this work, we attempt to bridge this gap by formulating and solving the Stochastic On-time Arrival (SOTA) problem for transit networks, which provides a transit routing policy that maximizes the probability of reaching the destination within a given travel time budget in stochastic transit networks. \par

A vast majority of previous studies on routing problems in stochastic transportation networks aim to identify the route or adaptive policy with the least expected travel time (LET) \citep{loui1983optimal,hall1986fastest, polychronopoulos1996stochastic,miller2000least,waller2002online,fan2005shortest,gao2006optimal,huang2012optimal, yang2014constraint, chen2014reliable}. While the expected travel time is a natural optimality metric, there exists a variety of situations where it is not sufficient, and tail statistics must be considered; for instance, travelers who want to catch a flight are more concerned with arriving on time with a high probability rather than with minimizing their expected travel time \citep{yang2017optimizing}. To account for these contexts, other formulations that consider a reliable optimal path have been proposed, starting with~\cite{frank1969shortest}. In this formulation, the goal is to find the path that maximizes the probability of realizing a travel time smaller than some desired time budget. This definition is subsequently extended to the online context in~\cite{fan2005arriving}, and is commonly referred to as the Stochastic On-Time Arrival (SOTA) problem. In the SOTA problem, the weight of each network link is a random variable with a known probability density function that represents the travel time of the link. The solution to the SOTA problem is a routing policy that maximizes the probability of arriving at the destination within a specified time budget. Here the routing policy is an adaptive solution that determines the routing decision at each node based on the realization of the travel time experienced en-route up to that point. \par

Given the complexity of the SOTA problem, significant efforts have been made to design efficient solution algorithms. \cite{fan2006optimal} formulate it as a dynamic programming problem and use a standard successive approximation (SA) procedure. This approach, however, has no finite bound on the maximum number of iterations needed for convergence in networks with loops. In~\cite{samaranayake2012speedup,samaranayake2012tractable}, the authors propose a label-setting algorithm to resolve this issue, exploiting the fact that there is always a non-zero minimum realizable travel time on each link in road networks. To further reduce the computation time, \cite{sabran2014precomputation} modify two deterministic shortest path preprocessing techniques: reach~\citep{gutman2004reach} and arc-flags~\citep{bauer2009sharc, hilger2009fast}, and apply them to the SOTA problem. Other studies \citep{nikolova2006stochastic, nie2009shortest, parmentier2014stochastic, niknami2016tractable} explore computationally efficient solution strategies for providing reliability guarantees in stochastic shortest path problems, where a fixed route (as opposed to a policy) is desired. \par

In transit networks, the problem of determining a set of reliable travel decisions is significantly more complex than in road networks. In this context, the traveler is not in control of the transit line(s) that they may travel on to reach the destination, and may have to make a number of complex decisions regarding which transit line(s) to take. This complexity arises from the following characteristics of transit systems.  


\begin{itemize}
\item \textbf{Uncertainty of arrival times/headways:} While the uncertainty in road networks is limited to the travel time, in transit networks one also needs to consider the uncertainty of arrival times and headways. Whether to take a transit line that arrives at a station or wait for a potentially faster service that is yet to arrive depends on the probabilistic trade-off between the extra waiting time and potential travel time savings. Therefore, computing the solution also requires modeling and solving for the headway distributions. 

\item \textbf{Combinatorial choice set:} 
In road networks, it never pays off to idle at a node and delay the departure from the node in hopes of improving the probability of arriving at the destination on time. However, in transit networks, it may be advantageous to not take the first transit line that arrives at the station (that can get you to your destination) and wait for a better option (e.g. wait for an express train or a more direct bus line). Therefore, a passenger needs to choose between multiple transit lines sharing segments of routes, and the decision regarding which transit line to board at a station depends on the unknown future arrival order of the candidate transit lines, which is an exponentially large set in the number of candidate (feasible) lines at each station, and leads to a combinatorial choice set.


\end{itemize}

Previous research on transit routing problems typically simplifies the problem by assuming that passengers board the first arriving transit service from an attractive line set that is precomputed to minimize the expected total travel time \citep{spiess1989optimal,cominetti2001common, nonner2014shortest, li2015finding}. Under this assumption, the passenger is committed to a transit line set. However, the traveler does not need to make a boarding decision until a transit vehicle is about to depart (\cite{hickman1997transit}), and it can be meaningful to adapt decisions based on information learned (uncertainties that are realized) during the trip. Fortunately, the role of online information has been noted recently in multiple studies. \cite{gentile2005route} propose a frequency-based assignment model under the assumption that the arrival time of the next transit vehicle is available once the passenger arrives at a transit station and demonstrate the potential benefits of utilizing such information. Instead of assuming full information, \cite{chen2015optimal} propose a routing strategy in a transit network with partial online information at the stations. The partial online information represents that the arrival time of the incoming transit vehicles is available only for a subset of the candidate transit lines. \cite{oliker2018frequency} develop a transit assignment model which considers two types of available information (partial information and full information), and demonstrate the impact of online information on assignment results. In our work, we consider the routing problem in transit networks as a fully online problem. The online information setting mentioned above can be easily adopted into our framework by changing the travel time and headway distributions accordingly. \par

The aim of this article is to formulate the SOTA problem for transit networks and develop efficient algorithms to solve it in this setting. The specific contributions of the work include: \par

\begin{itemize}
\item The formulation of the SOTA problem for transit networks, including a general network structure for stochastic transit networks and a decision-making model. Both the waiting time for each transit line and the travel time on each link are assumed to be random variables with known probability density functions. The solution is an adaptive policy that fully considers the transit specific characteristics of the problem mentioned above. \par

\item The design of a dynamic programming (DP) algorithm, which is pseudo-polynomial in the number of transit stations and time budget, and exponential in the number of transit lines at each station, which is practically a small number. To reduce the search space, we develop a definition of transit line dominance and present methods to identify this transit line dominance, which significantly decreases the computation time in our numerical experiments. \par

\item Extensive experiments in a synthetic network and in the Chicago transit network, which show the potential for solving this problem in a real-time route planning application setting. We also propose a general procedure to generate travel time, headway, waiting time distributions from General Transit Feed Specification (GTFS) data. \par
\end{itemize}

The rest of this article is organized as follows. In Section~\ref{sec:formulation}, we formulate the mathematical problem, in particular, we describe the network model and the underlying decision-making framework. In Section~\ref{sec: solve}, we introduce a dynamic programming based approach to solve the problem, as well as the complexity analysis of the algorithm. In Section~\ref{sec: space}, we provide a series of algorithmic techniques to reduce the search space. Section~\ref{sec:experiments} consists of numerical results on the computational performance and practical efficiency of the model and algorithms introduced in this work, both in a synthetic network and in the Chicago transit network. Finally, in Section~\ref{sec:conclusion}, we provide closing remarks and discuss possible directions for future research. \par

\section{Problem formulation.}\label{sec:formulation}
\subsection{Problem description}
In this section, we extend the SOTA problem definition to the context of transit networks. In this setting, as mentioned previously, two types of random variables need to be considered: the travel time between two transit stations and the waiting time for each transit line arrival at each transit station. The objective of the SOTA problem for transit networks is to find the routing (boarding, alighting and transferring) policy that maximizes the probability of arriving at the destination within a time budget. For conciseness, in the remainder of the article, we use \emph{utility} to represent \emph{the probability of arriving at the destination within the remaining time budget}, but note that the framework can be extended to other functions, and to robust settings \citep{flajolet2017robust}. The routing policy here includes the boarding, alighting and transferring decisions at each transit station in different situations (time already spent waiting at a station, arrival order of transit services, etc). \\

\noindent The following modeling assumptions are made. 
\begin{myAsum}
Passengers' arrivals are independent of the transit schedule.
\end{myAsum}
This assumption models a situation where the transit service is frequency-based (not schedule-based) and passengers do not adjust their specific departure times based on the transit schedule. This is a standard assumption made in the literature \citep{gentile2005route,li2015finding}. In modern services, transit vehicles' positions may be published in real-time (even for frequency-based services), and thus passenger arrivals may be correlated with the bus schedule. The formulation can be adapted to account for this by changing the travel time distribution and waiting time distribution accordingly. \par

\begin{myAsum}\label{asum: once}
Only the first arrival from each transit line (after the passenger enters the station) is considered as a candidate in the choice set. 
\end{myAsum}
This assumption is made for both modeling and algorithmic reasons:
\begin{itemize}
\item Guiding the passenger to ignore the first arrival from a particular line, but board the second arrival from that same transit line makes the routing direction confusing; 
\item Without this assumption, it is possible for the passenger to have to make an infinite number of decisions in the theoretical worst case (albeit with vanishing probability).
\end{itemize}

The passenger may not be able to board certain transit services (especially during rush hour or after a major event), if the transit service is full and has no remaining capacity. However, since we do not have data about the occupancy of each transit vehicle, the capacity of each transit service is not considered in our study. Note that this is not a disadvantage of our framework because given the data for the occupancy of transit vehicles, we should be able to define the waiting time to be the waiting time for the first arrival transit service that the passenger can board, and our framework can then be adapted to use these distributions without any changes.

\begin{myAsum}\label{asum: nosame}
No two transit services can arrive at the transit station at the same time. 
\end{myAsum}
Without loss of generality, two events never occur at exactly the same time in the continuous setting. In the subsequent discretized form of the problem, this assumption translates to no two transit services arriving at a transit station during the same discretized time interval\footnote{While the optimality of the solution relies on this assumption, if two lines arrive in the same time interval in practice, the algorithm can closely approximate the solution by considering the two cases of one arriving before the other.}. 

\subsection{Transit network representation} \label{network_structure}
We consider a directed graph $G(V, E)$, in which $V$ is the set of nodes and $E$ is the set of links. We model the transit stations with three types of nodes:
\begin{itemize}
\item \textbf{Station nodes:} A station node denoted by $S_y^X$ represents a passenger waiting at station $y$ for the set of candidate transit lines $X$.
\item \textbf{Arrival nodes:} An arrival node denoted by $A_y^{i,X}$ represents transit line $i$ arriving first at station $y$ among all the candidate lines in $X \cup \{i\}$, and transit lines $j \in X$ having not arrived yet, since the passenger arrives at the station.
\item \textbf{Line nodes:} A line node denoted by $L_y^i$ represents the passenger boarding transit line $i$ at station $y$.
\end{itemize}

\begin{figure}[htb]
\centering
\includegraphics[width=0.8\textwidth]{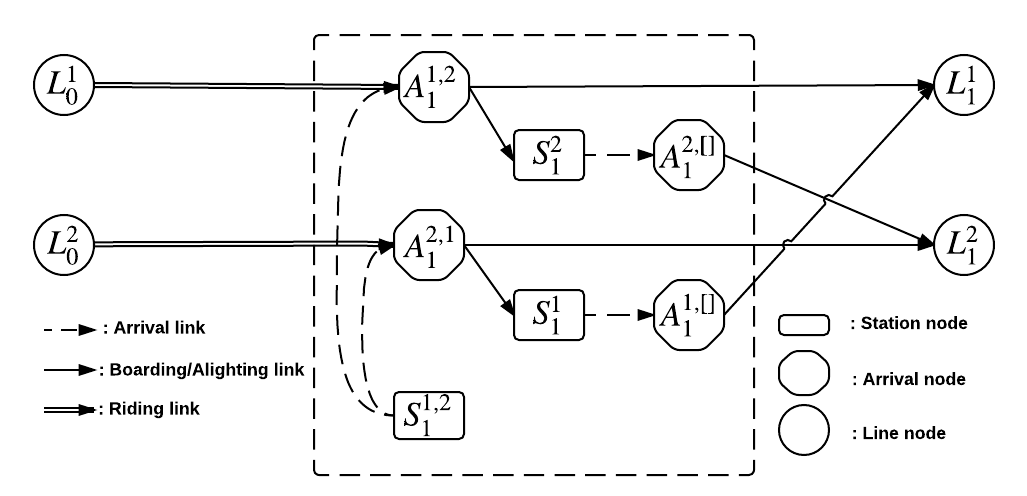}
\caption{\label{fig:network} Network representation of a transit station served by two transit lines. All the links and nodes within the dashed line rectangle are used to model the decision-making at a single physical station. }
\end{figure}

The set of station nodes, arrival nodes and line nodes are $V_S$, $V_A$ and $V_L$ respectively. Without loss of generality, the OD is selected from the station nodes to model a passenger starting their trip from a station node in the waiting state. Give a passenger in a waiting state at a station node with $m$ candidate transit lines, it is possible for any of the $m$ lines to arrive first. In each of these situations, the passenger either boards the transit line that arrives or continues waiting for other candidate lines. Passengers make decisions that maximize their utility. \\

We categorize links as follows: 
\begin{itemize}
\item \textbf{Arrival links:} Links from station nodes to arrival nodes.
\item \textbf{Riding links:} Links from line nodes to arrival nodes.
\item \textbf{Boarding links:} Links from arrival nodes to line nodes.
\item \textbf{Alighting links:} Links from arrival nodes to station nodes.
\end{itemize}

Figure~\ref{fig:network} illustrates the network representation of the decision-making process at a transit station with two transit lines. A passenger starting the trip physically at station 1 starts the trip at station node $S_1^{1,2}$ in the model. If transit line $1$ arrives first, the passenger moves to the corresponding arrival node $A_1^{1,2}$. The passenger has to decide to either board this transit line or continue to wait for line $2$. If the passenger continues waiting, the passenger moves to station node $S_1^2$. Otherwise, the passenger moves to line node $L_1^1$. All the links and nodes within the dashed line rectangle are used to model the decision-making at a single physical station. This network model allows for transfers. For instance, a passenger having boarded line 1 at station $S_0$ (the station preceding station $S_1$) is able to get to station node $S_1^2$ and wait for the next arrival of line $2$. \par

The link costs in the network are defined as follows. A riding link $(i,j)$ is associated with a travel time distribution $p_{i, j}(\cdot)$. A boarding link has no cost, since it refers to an instantaneous event. An arrival link is associated with a waiting time distribution $w_{y}^j(\theta,r)$, characterizing the probability density for the waiting time being $\theta$ for transit line $j$ at station $y$, given that the passenger has already waited at the station for $r$ units of time. Note that $\theta$ is the waiting time on top of $r$, and models transit line $j$ arriving at station $y$ after a total waiting time of $(r+\theta)$ since the passenger arrives at the station. Therefore, $w_{y}^j (\theta, 0)$ is the original waiting time distribution when the passenger arrives at the station. Research has shown that empirical headway data fit better with Loglogistic, Gamma and Erlang distributions than the exponential distribution \citep{li2015finding}. Therefore, we do not assume that the waiting time distribution is memoryless (as some other studies do), and use the following rule to normalize the waiting time distribution given $r$:

\begin{equation}\label{normalize}
w_{y}^j (\theta,r)=\frac{w_{y}^j (r+\theta, 0)}{1-\int_0^{r} w_{y}^j (\alpha, 0) d\alpha}\quad 0\leq \theta\leq T-r, 0< r\leq T
\end{equation}
with $T$ being the total time budget when the passenger begins the trip. \par

In Section~\ref{SOTAtransit}, we discretize the time-space to solve the problem numerically. In the discretized space, we force the waiting time to be at least one unit of the discretized time interval. Therefore, the waiting time equation reads as follows. 

\begin{equation}
w_{y}^j (\theta,r)=\frac{w_{y}^j (r+\theta, 0)}{1-\sum_0^{r} w_{y}^j (\alpha, 0) d\alpha}\quad 1\leq \theta\leq T-r, 0< r\leq T
\end{equation}

Although the waiting time distribution is a 2-D array, as explained in Section 3.2, we can save computation time in practice by precomputing and storing $1-\int_0^{r} w_{y}^j (\alpha, 0) d\alpha$ for each $r$ after we discretize the time-space. \par

\subsection{The SOTA problem for transit networks}\label{SOTAtransit} 
In this section, we describe how to compute the utility functions of the passenger at the different types of nodes in the model. The utility at a node $i$ is a function of the remaining time budget $t$ when the passenger arrives at the node, denoted by $u_i(t)$. The utility function at the destination node $D$ is:
    $$u_D (t)=1 \quad 0\leq t\leq T$$ since the passenger has completed the trip when at node $D$. 
\subsubsection{Line nodes}
Recall from Figure~\ref{fig:network} that a line node only contains an outgoing edge to an arrival node. Assume that the line node we are considering is $i$ and the following arrival node is $j$. The travel time on link $(i, j)$ is a random variable $\theta$, and the remaining time budget at node $j$ is $t - \theta$. Therefore, the utility at line node $i$ is a function of the remaining time budget $t$ when the passenger arrives at the node, denoted by $u_i(t)$. 
\begin{myDef}\label{def:line_node}
    The utility function at a line node $i$ can be computed as follows.
    $$u_i (t)= \mathop{\mathbb{E}}_{\theta} (u_{j}(t - \theta)) = \int_0^t p_{i, j}(\theta)\cdot u_j (t-\theta) d\theta, \forall i\in V_L, (i,j)\in E, 0\leq t\leq T$$ where $j\in V_A$ is the subsequent arrival node of line node $i$.
\end{myDef}
\noindent This utility function is analogous to the standard utility function of the SOTA problem for road networks.

\subsubsection{Arrival nodes}
All the routing decisions occur at arrival nodes, since the passenger is faced with the decision of either boarding the transit line that arrives first (selecting the subsequent line node) or continuing to wait for the other candidate transit lines (selecting the corresponding station node). Let $A_{y}^{i, X}$ be the arrival node of interest, and $u_{A_{y}^{i, X}} (t, r)$ be the utility when the passenger has a remaining time budget $t$ and has already waited at the station for $r$ units of time. We call the tuple $(t, r)$ the passenger state. Since the passengers aim to maximize the utility, the utility at an arrival node is the maximum of the utilities at the subsequent line node and station nodes.

\begin{myDef}\label{def:arrivalnode}
The utility function at an arrival node $A_y^{i,X}$ with passenger state $(t, r)$ is:
$$u_{A_y^{i,X}} (t,r)=\max_{ j\in V_L \mid (A_y^{i,X},j)\in E;\  S_y^X\in V_S \mid (A_y^{i,X},S_y^X)\in E} \{u_j (t),u_{S_y^X} (t,r)\}$$
\end{myDef}

\subsubsection{Station nodes}\label{station_nodes}
Station nodes are followed only by arrival nodes, and the utility at a station node is the expectation of the utility at the following arrival nodes. An arrival node is defined by the first arriving transit line and the corresponding passenger state. Let $(i, \theta)$ be a random event which represents that transit line $i$ is the first arriving transit line, and the waiting time for it is $\theta$. Assume that $S_y^{X}$ is the station node of interest, and that we want to compute $u_{S_y^X} (t,r)$. For each $0\leq \theta\leq t$, the probability of the transit line $i$ arriving the first after waiting $\theta$ units of time is $w_{y}^i(\theta,r)\cdot \prod\limits_{j\in X\backslash i}(1-\int_0^\theta w_{y}^j(\alpha,r)~d \alpha)$. In this case, the passenger has waited $\theta+r$ units of time in total, and the utility at this arrival node is $u_{A_y^{i,X\backslash i}} (t-\theta,\theta+r)$.
\begin{myDef}
    The utility function at a station node is:
    \begin{equation*}
        \begin{split}
        u_{S_y^X} (t,r) &= \mathop{\mathbb{E}}_{(i, \theta)} (u_{A_y^{i,X\backslash i}} (t-\theta,\theta+r)) \\
         &= \sum_{i\in X} \left(\int_0^t w_{y}^i(\theta,r)\cdot\prod_{j\in X\backslash i}(1-\int_0^\theta w_{y}^j(\alpha,r)d\alpha) \cdot u_{A_y^{i,X\backslash i}} (t-\theta,\theta+r)d\theta\right)
        \end{split}
    \end{equation*}
\end{myDef}
\noindent  This utility is the weighted average of the utility at the corresponding arrival nodes. \par

The integrals in the above equations do not have closed form expressions and cannot be solved analytically. Therefore, they need to be integrated numerically by discretizing the time horizon into small intervals. In the discretized model, we assume that no two transit lines can arrive at the same transit station at the same discretized time interval. As a consequence, the solution not guaranteed to be optimal, since there is a small (non-zero) probability that two lines might arrive at the same time interval regardless of how small the discretization is. However, we can set the length of time intervals to be a small number in practice\footnote{In the experiments, we set the time interval length to be 15 seconds. Using this time interval length, the average probability of multiple buses arriving at the same time interval for each station is about 0.6\% in the Chicago transit network we use in numerical experiments. For a particular OD, this probability is usually lower since the candidate bus line set is a subset.}. \par 

Given a fixed time discretization, the discrete form of the utility function at station node reads as follows. 
\begin{myDef}\label{def:stationnode}
    The discrete form of the utility function at station nodes is:
    $$u_{S_y^X} (t,r)=\sum_{i\in X} \left(\sum_{\theta=1}^t w_{y}^i(\theta,r)\prod_{j\in X\backslash i}(1-\sum_{\alpha=0}^\theta w_{y}^j(\alpha,r)) \cdot u_{A_y^{i,X\backslash i}} (t-\theta, \theta+r)\right)$$.
\end{myDef}
\noindent The utility function for other types of nodes can be discretized similarly. \par

\begin{myEx}
Here we show a simple example that shows the sophisticated decision-making process that the model allows. Assume that the passenger is waiting at a station with three candidate transit lines. Table~\ref{tab: ex1-distribution} is the waiting time distribution for each transit line, and the utility of taking the transit line corresponding to the waiting time. Figure~\ref{boarding} shows the optimal decision based on the arrival order. We can observe that: (1) The transit line that first arrives is not always the optimal choice. When transit line 3 arrives at time 2, it is the first transit service that arrives. However, it is optimal for the passenger to not board, and continue to wait; (2) Even when it is optimal to not board a transit line at an arrival event, the event may impact the future optimal decisions, as illustrated by the differing policies following the arrival or non-arrival of transit line 3 at time 2.
\end{myEx}

\begin{table}[H]
\centering
\caption{The waiting time and utility distribution for example 1. }
\label{tab: ex1-distribution}
\resizebox{.55\textwidth}{!}{
\begin{tabular}{c|c|c|c}
\hline
Transit line ID    & Waiting time & Probability & Utility \\ \hline
\multirow{3}{*}{1} & 1            & 0.05        & 0.90    \\ \cline{2-4} 
                   & 3            & 0.05        & 0.80    \\ \cline{2-4} 
                   & 10           & 0.90        & 0.00    \\ \hline
\multirow{2}{*}{2} & 5            & 0.90        & 0.85    \\ \cline{2-4} 
                   & 15           & 0.10        & 0.00    \\ \hline
\multirow{2}{*}{3} & 2            & 0.50        & 0.70    \\ \cline{2-4} 
                   & 6            & 0.50        & 0.60    \\ \hline
\end{tabular}
}
\end{table}

\begin{figure}[H]
\centering
\includegraphics[width=0.9\textwidth]{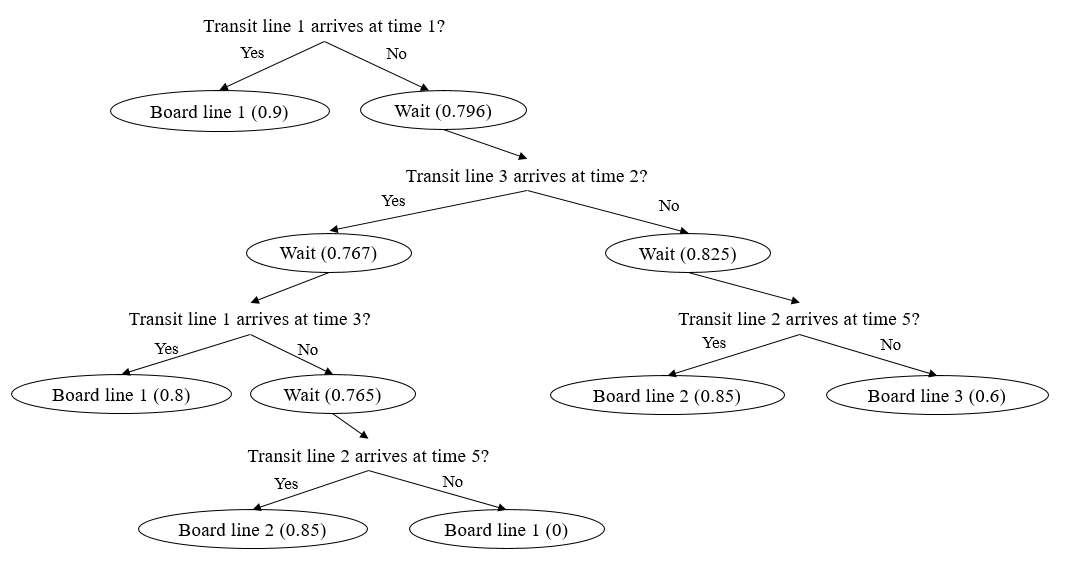}
\caption{\label{boarding} The optimal boarding decision corresponding to the transit line arriving order for example 1. The number in each decision node (board and wait) represents the decision's utility. Note that arrival events for which boarding is never optimal are ignored in the figure (e.g. the arrival of transit line 1 at time 3 following the non-arrival of transit line 3 at time 2).}
\end{figure}

\section{Solving the SOTA problem for transit networks}\label{sec: solve}
In this section, we describe how to solve the discrete time version of the SOTA problem for transit networks. We first show how the problem can still (even in the transit setting) be solved using a dynamic programming approach, then present a complexity analysis of the algorithm, and finally in Section~\ref{sec: space} describe a number of search space pruning methods designed to make the problem tractable in practice.

\subsection{Dynamic programming approach}
In \cite{samaranayake2012tractable}, a label-setting algorithm for the SOTA problem is developed by exploiting the fact that each link in a road network has a positive minimum realizable travel time. This fact guarantees that there will be no loops in the network with zero travel time, and thus allows for a dynamic programming approach for solving the problem. However, in the transit network representation proposed in the previous section, there are links with zero minimum realizable travel time. 

\begin{myCla} \label{c1}
In the transit network representation introduced in Section~\ref{sec:formulation}, there is no loop with a zero realizable travel time.
\end{myCla}
\begin{proof}
Only the links starting from the arrival nodes can have zero minimum realizable time. To form a loop, the passenger has to first get to an arrival node from some other type of node, which is not an arrival node. The links originating from all other types of nodes have a positive minimum realizable travel time, so a zero travel time loop is not possible. 
\end{proof}

Therefore, we can still use a dynamic programming approach to solve the SOTA problem for transit networks. Each sub-problem in the dynamic program can be defined by $(i, t, r)$ where $i$ is the node that we consider, $t$ and $r$ have the same definition as in Section 2. Let $\text{OPT}(i, t, r)$ denote the utility at node $i$ when the passenger has a time budget of $t$ and has waited at the station for $r$ units of time. Note that the $\text{OPT}(i, t, r)$ satisfies the Bellman's Principle of Optimality, i.e., the remaining decisions only depend on the current state and are not related to the past states and decisions. According to the definitions in Section 2, we claim that $\text{OPT}(i, t, r)$ satisfies the following recurrence relation: \par

\begin{equation}\label{equ: recurrence}
\resizebox{\textwidth}{!}{%
$\text{OPT}(i, t, r)=\left\{
\begin{aligned}
& 1 & \quad \text{if } i \text{ is the destination} \\
& 0 & \quad \text{if } t < 0 \\
& \sum\limits_0^t p_{i, j}(\theta)\cdot \text{OPT}(j, t-\theta, 0) & \quad \text{if } i \text{ is a line node and } (i, j) \in E \\
& \max_{ j\in V_L \mid (A_y^{i,X},j)\in E;\ S_y^X\in V_S \mid (A_y^{i,X},S_y^X)\in E} \{\text{OPT}(j, t, 0), \text{OPT}(S_{y}^{X}, t, r)\} & \quad \text{if } i \text{ is arrival node } A_{y}^{i, X} \\
& \sum_{i\in X} \left(\sum_{\theta=1}^t w_{y}^i(\theta,r)\prod_{j\in X\backslash i}(1-\sum_{\alpha=0}^\theta w_{y}^j(\alpha,r)) \cdot \text{OPT}(A_y^{i,X\backslash i}, t-\theta, \theta+r)\right) & \quad \text{if } i \text{ is station node } S_{y}^{X} \\
\end{aligned}
\right.$}
\end{equation}

The initial problem we wish to solve is given by $(O, T, 0)$, i.e., the passenger is at the origin node $O$ with time budget $T$ and has waited for zero units of time. The algorithm first initializes the utility function corresponding to the destination node to 1 and computes all
the normalized waiting time distributions that are needed in the subsequent computation. Then, the utility at each node is updated in a dynamic programming fashion according to the recurrence relation given in Equation~\ref{equ: recurrence}.\par \par



    
    

\subsection{Complexity analysis} 
The runtime of the algorithm is pseudo-polynomial in the number of stations in the transit network and time budget, and exponential in the number of transit lines at any station. We first provide the number of nodes of all three types. Then, we analyze the time complexity for computing the utility on each type of node. 

\begin{myCla} \label{c: n_nodes}
For a station with $m$ transit lines, there are i) $m$ line nodes, ii) $\sum\limits_{i=1}^{m} C_{m}^i$ station nodes, and iii) $\sum\limits_{i=1}^{m} i\cdot C_{m}^i$ arrival nodes.
\end{myCla}
\begin{proof}
i) For every transit line, there is a corresponding line node. Therefore, there are $m$ line nodes in total. ii) For each combination of transit lines representing the non-empty set of lines yet to arrive, there is a station node. Therefore, there are $\sum\limits_{i=1}^{m} C_{m}^i$ station nodes in total. iii) For each station node, any transit line can be the line arriving the first and yield an associated arrival node. Therefore, there are $\sum\limits_{i=1}^{m} i\cdot C_{m}^i$ arrival nodes. 
\end{proof}

We now analyze the time complexity of the algorithm for a transit network with a station set $Y$ and each station has no more than $m$ transit lines. Let $M_y$ be the set of transit lines at station $y$. The algorithm first computes all the normalized waiting time distributions that are needed in the subsequent computation, which takes $O(|Y|\cdot m\cdot T^3)$ time based on Equation~\ref{normalize}. In our implementation, we use $O(|Y|\cdot m \cdot T)$ memory to store $1 - \sum_{\alpha=0}^{\theta} w_{y}^j(\alpha,0), \forall \theta \leq T, y \in Y, j \in M_y$, which decreases the time complexity to $O(|Y|\cdot m\cdot T^2)$. 

\begin{myCla} \label{c: u_complexity}
For a station set $Y$ in which each station has no more than $m$ transit lines, the time complexity of computing the utility functions for a time budget of $T$ is:
\begin{itemize}
\item $O(|Y|\cdot m\cdot 2^{m-1}\cdot T^2)$ for all arrival nodes,
\item $O(|Y|\cdot (m^{2} - m) \cdot 2^{m-2}\cdot T^3)$ for all station nodes.
\end{itemize}
\end{myCla}
\begin{proof}
Based on Definition~\ref{def:arrivalnode}, we need to update the utility functions at each arrival node for each possible remaining time budget $t$ and each waiting time $r$. In addition, based on Claim \ref{c: n_nodes}, there are $\sum\limits_{i=1}^{m} i\cdot C_{m}^i$ arrival nodes for a station with $m$ transit lines, so it takes $O(|Y|\cdot \sum\limits_{i=1}^{m} i\cdot C_{m}^i\cdot T^2)=O(|Y|\cdot m\cdot 2^{m-1}\cdot T^2)$ time to compute the utility at all the arrival nodes.

Assume that the station node of interest is $S_y^{X}$. To compute $U_{S_y^{X}}(t, r)$, for each transit line $i\in X$, we need to compute the probability of line $i$ arriving among all lines in $X$, which is $\sum_{\theta=1}^t w_{y}^i(\theta,r)\prod_{j\in X\backslash i}(1-\sum_{\alpha=0}^\theta w_{y}^j(\alpha,r))$ according to Definition~\ref{def:stationnode}. Therefore, for each $i\in X$ and passenger state $(t, r)$, we need $O((|X| - 1)\cdot T^2)$ to compute the above probabilities. Based on Definition~\ref{def:stationnode}, we need to update the utility functions at each station node for each possible remaining time budget $t$ and each waiting time $r$. Consequently, we need $O(|X|\cdot (|X|-1) \cdot T^{4})$ to update the utility functions for each station node in total. According to Claim \ref{c: n_nodes}, there are $\sum\limits_{i=1}^{m} C_{m}^i$ station nodes for a station with $m$ transit lines. Therefore, it takes $O(|Y|\cdot \sum\limits_{i=1}^{m} (C_{m}^i \cdot i \cdot (i-1) )\cdot T^4)=O(|Y|\cdot (m^{2} - m) \cdot 2^{m-2}\cdot T^4)$ time to compute the utility at all the station nodes, without re-using any information. We find that the probability of transit line $j$ arriving after $\alpha$ time given that the passenger has already waited $r$ time at station $y$, which appears in Definition \ref{def:stationnode} reads: $1-\sum_{\alpha=0}^\theta w_{y}^j(\alpha,r) = \frac{1 - \sum_{\alpha=0}^{r + \theta} w_{y}^j(\alpha,0)}{1 - \sum_{\alpha=0}^{r} w_{y}^j(\alpha,0)}$. In our implementation, we use $O(|Y|\cdot m \cdot T)$ memory to store $1 - \sum_{\alpha=0}^{\theta} w_{y}^j(\alpha,0), \forall \theta \leq T, y \in Y, j \in M_y$. Then, the time complexity for computing utility for all station nodes becomes $O(|Y|\cdot m^{2} \cdot 2^m\cdot T^3)$. 
\end{proof}

In summary, the time complexity of the algorithm is $O(|Y|\cdot m^2 \cdot 2^m\cdot T^3)$. Considering that there is a constant maximum number of transit lines at a station, the time complexity is $O(|Y|\cdot T^3)$, which leads to a pseudo-polynomial time algorithm in $|Y|$ and $T$. 

\section{Search space reduction}\label{sec: space}
Although the DP approach is a pseudo-polynomial time algorithm in $|Y|$ and $T$, the computation time can still be high in large-scale networks when $T$ is large. Therefore, we propose some search space reduction techniques to further decrease the computation time in practice. \par

\subsection{Eliminate the infeasible paths}
The simplest pruning technique we employ is to eliminate infeasible paths, i.e., paths that have zero probability of being used based on the minimum realizable travel time on each link. This pruning can be performed by simply running a shortest path search on a modified graph where the link weight is the minimum realizable travel time, to find the shortest path distance from each station to the destination. Assume that the minimum realizable travel time from station $y$ to the destination is $\alpha_y$. The utility for any node at station $y$ given a time budget smaller than $\alpha_y$ will be zero. 
\par

This method is similar to the pruning method in \cite{samaranayake2012tractable} for road networks. Since the complexity of the shortest path algorithm is dominated by the complexity of the SOTA problem, the cost of the pruning method is negligible compared to the total computation time. \par

\subsection{Search space reduction using transit line dominance}
To compute the utility at an arrival node, we need to compare the utility of boarding the transit service (line node) and continuing to wait (station node). In this section, we propose a definition of transit line dominance, and a set of computationally efficient conditions to check for such dominance. This allows us to save the computation for the utility of the corresponding station node in the case of a dominating line node. \par


\subsubsection{Dominance definition and properties} \label{sec: relation}

\begin{myDef}\label{def:dom}
Assume that the passenger is at node $A_{y}^{i, X}$ with passenger state $(t, r)$. Transit line $i$ dominates $X$ if $u_{L_{y}^{i}}(t) \geq u_{S_{y}^{X\backslash i}} (t, r)$. We use $i\succeq X$ to represent that $i$ dominates $X$, and $i\prec X$ to represent that $i$ does not dominate $X$.
\end{myDef}

According to the definition, if $i \succeq X$, the passenger should board the transit line $i$; If $i \prec X$, the passenger should continue to wait for the transit lines in $X$. We now propose a useful claim that will be used in the proof of the dominance properties that we propose later. \par

\begin{myCla}\label{c: addline}
Assume that the passenger is at station $y$ with passenger state $(t, r)$. $u_{S_y^{X}}(t, r) \leq u_{S_y^{X^\prime}}(t, r), \forall X^\prime \supset X$.
\end{myCla}

\begin{proof}
The passengers in our system are utility maximizers. Therefore, more transit line choices will make the utility of the station node increase or remain the same.
\end{proof}

\begin{myPro}\label{cor: domsub}
Assume that the passenger is at station $y$ with passenger state $(t, r)$. If $i \succeq X$, then $i \succeq X^\prime, \forall X^\prime \subset X$. Correspondingly, if $i \prec X$, then $i \prec X^\prime, \forall X^\prime \supset X$. 
\end{myPro}

\begin{proof}
We prove the first part of the claim. The second part can be proved using similar reasoning. Assume that the passenger is at station $y$. According to Definition~\ref{def:dom}, $u_{L_y^{i}}(t) \geq u_{S_y^{X}}(t, r)$ if $i \succeq X$. In addition, for any $X^\prime \subset X$, $u_{S_y^{X^\prime}}(t, r) \leq u_{S_y^{X}}(t, r)$ based on Claim~\ref{c: addline}. Therefore, $u_{L_y^{i}}(t) \geq u_{S_y^{X^\prime}}(t, r)$, and thus $i\succeq X \implies i \succeq X^\prime, \forall X^\prime \subset X$. 
\end{proof}

\begin{myPro}
Assume that the passenger is at station $y$ with passenger state $(t, r)$. If $u_{L_y^{i}}(t) \geq u_{L_y^{j}}(t)$, then: (1) $i \succeq X$ for any $X$ such that $j\succeq X$; (2) $j \prec X^\prime$ for any $X^\prime$ such that $i \prec X^\prime$. 
\end{myPro}

\begin{proof}
For (1), if $j \succeq X$, then $u_{L_y^{j}}(t) \geq u_{S_y^{X}}(t, r)$ according to the definition of transit line dominance. Since $u_{L_y^{i}}(t) \geq u_{L_y^{j}}(t)$, we derive $u_{L_y^{i}}(t) \geq u_{S_y^{X}}(t, r)$, which proves that $i \succeq X$. \par

For (2), as $i \prec X^\prime$, $u_{L_y^{i}}(t) \leq u_{S_y^{X}}(t, r)$. Since $u_{L_y^{i}}(t) \geq u_{L_y^{j}}(t)$, $u_{L_y^{j}}(t) \leq u_{S_y^{X}}(t, r)$, which represents $j \prec X^\prime$.
\end{proof}

According to the two properties above, we can infer transit line dominance based on the transit line dominance that we already know. Therefore, at an arrival node $A_y^{i, X}$ with passenger state $(t, r)$, if we can infer $i \succeq X$ from the known dominance, we can reduce the computation for $u_{S_y^{X}}(t, r)$. \par

\subsubsection{Dominance conditions}\label{sec: dom}
In this section, we present a series of conditions that can be assessed efficiently, and are practically useful to reduce the number of utility computations at station nodes. The conditions are considered at arrival nodes in order. \par


The first dominance condition is a subproblem pruning technique. If the condition is satisfied, we do not need to compute the utility at the corresponding station node at this specific passenger state. \par

\begin{myProp} \label{p: dom}
(Subproblem pruning) Assume that the passenger is at node $A_y^{i, X}$ with passenger state $(t, r)$. If $u_{L_y^i} (t)\geq \max\limits_{j\in X} \{ u_{L_y^j} (t-1)\}$, then $i\succeq X$. 
\end{myProp}

\begin{proof}
The utility function is a nondecreasing function with respect to the time budget, i.e., $u_{L_y^j} (t-1)\geq u_{L_y^j} (t^\prime), \forall t^\prime\leq t-1, \forall j \in X$. In addition, $u_{S_y^{X}}(t, r)$ is a weighted average of $u_{L_y^{j}}(t^\prime)$ for $t^\prime \leq t$ and $j \in X$ where the sum of the weight is not larger than 1. Therefore, if the condition is satisfied, $u_{L_y^i} (t) \geq u_{S_y^{X}}(t, r) \implies i \succeq X$.
\end{proof}

Proposition \ref{p: dom} needs $O(m)$ time to be checked. The intuition is as follows: Assume that we know all the transit lines in $X$ will arrive right after $i$ arrives. If boarding the transit line $i$ under this assumption has a larger utility than waiting for the rest, then boarding the transit line $i$ also has a larger utility than waiting for the rest without the assumption (i.e. the transit lines in $X$ might arrive later). \par

If Proposition~\ref{p: dom} cannot prove that $i \succeq X$, then we need to compute $u_{S_y^{X}}(t, r)$ explicitly. We show next that, in some cases, we can safely remove some transit lines in $X$ and only consider a subset of the candidate transit lines. Before we propose this candidate set pruning technique, we first provide the following definition of \textit{improving lines} that will be used later on. \par 
\begin{myDef} \label{def: improve}
Assume that the passenger is at station $y$ with passenger state $(t, r)$. Transit line $k$ improves line $j$ at station $y$ if $u_{S_y^{\{k, j\}}}(t, r) > u_{S_y^{\{j\}}}(t, r)$.
\end{myDef}

A sufficient and necessary condition for transit line $k$ improving line $j$ at station $y$ is given as follows. \par

\begin{myCla} \label{c6} 
Assume that the passenger is at node $A_y^{i, X}$ with passenger state $(t, r)$. Line $k$ improves line $j$ if there exists some state $(t-\gamma,r+\gamma)$ such that $u_{L_y^k} (t-\gamma)>\sum\limits_{\theta=1}^{t-\gamma} (w_{y}^j (\theta,r+\gamma)\cdot u_{L_y^j} (t-\gamma-\theta))$.
\end{myCla}

\begin{proof}
If there exists a passenger state $(t-\gamma,r+\gamma)$ satisfying the condition above in Claim~\ref{c6}, the passenger should board $k$ when it arrives at state $(t-\gamma, r+\gamma)$ instead of continuing to wait for $j$. Therefore, the utility at the station node with $\{k,j\}$ is larger than the utility at station node with only line $j$. 
\end{proof}

Now, we propose the candidate set pruning technique as follows. \par

\begin{myProp}
(Candidate set pruning) Assume that the passenger is at node $A_y^{i, X}$ with passenger state $(t, r)$. Let $Z = \{j \in X\mid u_{L_y^{i}(t)} < u_{L_y^j}(t - 1)\}$. When we compute $u_{S_y^{X}}(t, r)$, we can instead compute $u_{S_y^{X^\prime}}(t, r)$ where $X^\prime \subseteq X$ and $X^\prime = Z \cup \{j \in X \mid j \text{ improves at least one line in } Z\}$. 
\end{myProp}

\begin{proof}
Let $K = X \backslash X^\prime$. Since no line in $K$ improves the lines in $Z$, it is never better to take a line in $K$ than waiting for the lines in $Z$, i.e., $u_{S_y^{Z \cup K}}(t, r) = u_{S_y^{Z}}(t, r)$. Since $Z \subset X^\prime$, it is also never better to take the line in $K$ than waiting for the lines in $X^\prime$, which implies $u_{S_y^{X^\prime}}(t, r) = u_{S_y^{X}}(t, r)$. 
\end{proof}

We can consider the candidate set pruning technique as an extension of Proposition~\ref{p: dom} since $Z$ is obtained when we check the condition in Proposition~\ref{p: dom}. If $|Z| = 0$, the condition in Proposition~\ref{p: dom} is satisfied, $i \succeq X$. Otherwise, we prune the transit lines in $\{j \in X\backslash Z \mid j \text{ cannot improve any line in } Z\}$ before computing for the utility at the station node. \par

According to Definition~\ref{def:stationnode}, we iterate through the time intervals $\{z \mid 1 \leq z \leq t\}$ when we compute $u_{S_y^{X}}(t, r)$. Here $z$ is the waiting time for the first transit line arrival. In some cases, we can infer $i \succeq X$ when we finish the computation for an iteration $z < t$ and thus reduce the computation for the rest of the iterations. \par

\begin{myProp} \label{p: early_stop}
(Time interval pruning) Assume that the passenger is at $A_y^{i,X}$ with passenger state $(t,r)$. Let $u_{\text{max}}$ be the maximum utility among all subsequent arrival nodes of $S_{y}^{X}$ at passenger state $(t - z, r + z)$, i.e., $u_{\text{max}}(z)=\max\limits_{j\in X}\{{u_{A_y^{j,X\backslash j}}}(t-z,z+r)\}$. Let $p_{j, X}(\theta)$ be the probability that transit service $j$ is the first transit service arrives at the station among all transit lines in $X$, and it arrives at the station at passenger state $(t - \theta, r + \theta)$, i.e., $p_{j, X}(\theta) = w_{y}^j(\theta,r)\prod\limits_{k\in X\backslash j}(1-\sum\limits_{\alpha=1}^\theta w_{y}^k(\alpha,r))$. Let $u_{\text{sum}}(z)$ be the expected utility for the cases where there is one transit service in $X$ arriving at the station from passenger state $(t, r)$ to $(t - z, r + z)$, i.e., $u_{sum}(z)=\sum\limits_{j\in X} \sum\limits_{\theta=1}^z (p_{j, X}(\theta) \cdot u_{A_y^{j,X\backslash j}} (t-\theta,\theta+r))$. If there exists $z\leq t$ such that $u_{L_y^i} (t)\geq u_{sum}(z)+(1-\sum\limits_{j\in X} \sum\limits_{\theta=1}^z p_{j, X}(\theta))\cdot u_{max}(z)$, then $i\succeq X$.
\end{myProp}

\begin{proof}
If we can show that $u_{sum}(z)+(1-\sum\limits_{j\in X} \sum\limits_{\theta=1}^z p_{j, X}(\theta))\cdot u_{max}(z) \geq u_{S_y^{X}}(t, r)$, then the proposition is proved. According to Definition~\ref{def:stationnode}, $u_{S_y^{X}}(t, r) = u_{sum}(t)$. We can infer that $u_{S_y^{X}}(t, r) = u_{sum}(z) + \sum\limits_{j\in X} \sum\limits_{\theta=z + 1}^t (p_{j, X}(\theta) \cdot u_{A_y^{j,X\backslash j}} (t-\theta,\theta+r))$. Since the utility function is non decreasing, we know that $u_{\text{max}}(z) = \max\limits_{j \in X, z \leq \theta \leq t} \{{u_{A_y^{j,X\backslash j}}}(t-\theta,\theta+r)\}$. In addition, $1 - \sum\limits_{j\in X} \sum\limits_{\theta=1}^z p_{j, X}(\theta) \geq \sum\limits_{j\in X} \sum\limits_{\theta=z + 1}^t p_{j, X}(\theta)$. Therefore, we can conclude that the proposition is correct since $u_{S_y^{X}}(t, r) \leq u_{\text{sum}}(z)+(1-\sum\limits_{j\in X} \sum\limits_{\theta=1}^z p_{j, X}(\theta) \cdot u_{\text{max}}(z)$.
\end{proof}

Whenever we finish computing for a specific time interval $z = i$, we can check the condition in Proposition \ref{p: early_stop}. If the condition is not satisfied, we let $z = i + 1$ and continue the procedure. However, if the condition is satisfied, we know that $i\succeq X$ and thus can reduce the computation for $i + 1 \leq z \leq t$. \par

To sum up all the techniques discussed in this section and Section \ref{sec: relation}, we modify the procedure of computing the utility at arrival nodes as follows. Two dictionaries $dom$ and $nondom$ are used to record the transit line dominance and non-dominance. $update\_dom$ and $update\_nondom$ are two functions to update $dom$ and $nondom$ according to dominance properties in Section \ref{sec: relation}. Specifically, for any arrival node $A_y^{i,X}$, the pruning is done via the following set of steps. \par
\noindent\fbox{%
    \parbox{\textwidth}{%
Step 1: Check if $X$ is in $dom[i,t,r]$: if yes, return $u_{L_y^i}(t)$; Otherwise, go to Step 2.

Step 2: Check if $X$ is in $nondom[i,t,r]$: if yes, compute $u_{S_y^X}(t,r)$ and return it; Otherwise, go to Step 3.

Step 3: Check if $i\succeq X$ or $i\prec X$ using the pruning techniques. If $i\succeq X$, run $update\_dom$ and return $u_{L_y^i}(t)$; Otherwise, run function $update\_nondom$ and return $u_{S_y^X}(t,r)$.
    }%
}

We call the algorithm DP with dominance. Note that all the methods discussed retain the optimality of the solution.\par

\subsection{Heuristic rules}\label{sec:heur}
In this section, we propose three heuristic rules, which can be employed in the procedure of checking dominance, to further decrease the computation time. We will show in the numerical experiments that the heuristics can significantly reduce the computation time without much loss of the results accuracy. \par

\begin{myHeu}\label{h: dom}
Assume that the passenger is at node $A_y^{i,X}$ with passenger state $(t,r)$. Let $Z = \{j \in X\mid u_{L_y^{i}(t)} < u_{L_y^j}(t - 1)\}$, and $p = \prod\limits_{j \in Z} \sum\limits_{\theta = 1}^{t} I_{u_{L_y^{i}(t)} \geq u_{L_y^{j}(t - \theta)}} \cdot w_{y}^{j}(\theta, r)$ where $I_{u_{L_y^{i}(t)} \geq u_{L_y^{j}(t - \theta)}}$ is an indicator function. If $p \geq \epsilon$, we say that the passenger should board the transit line $i$. 
\end{myHeu}

This is an extension of Proposition~\ref{p: dom}. $\epsilon$ is a constant coefficient. $p$ is the probability of the case that boarding transit line $i$ has a larger utility than boarding any other single transit line in $X$. If $p$ is large, then only with a small probability, the realization of the waiting time for a transit line in $X$ can make its utility larger than boarding transit line $i$. \par

\begin{myHeu}\label{h: anyother}
Assume that the passenger is at node $A_y^{i, X}$ with passenger state $(t, r)$. If $u_{L_y^i} (t) \geq \sum\limits_{\theta=1}^{t} w_{y}^j (\theta,r)\cdot u_{L_y^j} (t-\theta), \forall j \in X$, we say that the passenger should board the transit line $i$. 
\end{myHeu}

In other words, the passenger should board transit line $i$ if boarding it has a larger utility than waiting for any single line in $X$. It is an approximation since $u_{S_y^{X}}(t, r) \geq u_{S_y^{\{j\}}}(t, r), \forall j\in X$ according to Claim~\ref{c: addline}. Therefore, it is possible that $u_{S_y^{X}}(t, r) > u_{L_y^i} (t)$. \par

\begin{myHeu}\label{h: timeinterval}
Assume that the passenger is at $A_y^{i,X}$ with passenger state $(t,r)$. Let $u_{max}(z)=\max\limits_{j\in X}\{{u_{A_y^{j,X\backslash j}}}(t-z,z+r)\}$, and $u_{sum}(z)=\sum\limits_{j\in X} (\sum\limits_{\theta=1}^z w_{y}^j(\theta,r)\prod\limits_{k\in X\backslash j}(1-\sum\limits_{\alpha=1}^\theta w_{y}^k(\alpha,r)) \cdot u_{A_y^{j,X\backslash j}} (t-\theta,\theta+r))$. If there exists $z\leq t$ such that $\beta \cdot u_{L_y^i} (t)\geqslant u_{sum}(z)+(1-\sum\limits_{j\in X} (\sum\limits_{\theta=1}^z w_{y}^j(\theta,r)\prod\limits_{k\in X\backslash j}(1-\sum\limits_{\alpha=1}^\theta w_{y}^k(\alpha,r))))\cdot u_{max}(z)$, then we say that the passenger should board transit line $i$.
\end{myHeu}

This is an extension of Proposition~\ref{p: early_stop}. The difference is that we add a constant relaxation coefficient $\beta$ ($\beta > 1$) in the condition. This represents that when we compute the utility of a station node, if $\beta$ times the utility of the line node is larger than the largest possible utility of the station node, we say that the passenger should board the transit line that arrives. $\beta$ in Heuristic~\ref{h: timeinterval} and $\epsilon$ in Heuristic~\ref{h: dom} are used to control the tradeoff between the results accuracy and computation performance. In the numerical experiments, we let $\beta = 1.25$ and $\epsilon = 0.75$. \par

In the following section, we present numerical results and compare the three versions of the solution algorithm introduced in this work: (1) DP; (2) DP with dominance; (3) DP with dominance and heuristics. \par

\section{Numerical experiments}\label{sec:experiments}
In this section, we present numerical experiments focused on illustrating the runtime performance of the various versions of the solution, as well as the practical value of using a SOTA policy in transit networks. All the algorithms are coded in Python $3.6$, and all tests are conducted on an AMD Ryzen processor computer (3.4 gigahertz, 16 gigabytes RAM). We first validate the algorithms and conduct a sensitivity analysis in a controlled setting using a synthetic transit network, and then test them on the Chicago transit network. 

\subsection{Estimating travel time and headway distributions}
To obtain the required network and transit line information, we develop a procedure for taking the input data (GTFS data) and generate all the processed data we need in the experiments, including travel time distributions, headway distribution, and waiting time distributions. The first step is to generate the travel time distributions for each transit line. Studies have shown that travel time distributions on road networks can be well approximated by a lognormal distribution (see \cite{emam2006using,hunter2009path,zou2014space}). Therefore, in our experiments, we model the travel time distributions using appropriately parameterized lognormal distributions. A lognormal distributed random variable $X$ is parameterized by two parameters $\mu$ and $\sigma$ that are, respectively, the mean and standard deviation of the variable's natural logarithm. For every two adjacent stations $S_1$ and $S_2$ in a transit line, we calibrate $\mu$ and $\sigma$ for the travel time between the two stations based on the following steps: \par

(1) Compute the minimum realizable travel time $t_m$ by dividing the distance between $S_1$ and $S_2$ by the corresponding speed limit. \par

(2) Let the mode of the lognormal distributed variable, i.e., the point of the maximum of the probability density function, be the difference between the scheduled departure time at $S_2$ and $S_1$. The mode of a lognormal distribued variable is $e^{\mu - \sigma^2}$. Therefore, if the schedule departure time at $S_1$ and $S_2$ are 8:45 and 9:00 and $t_m = 5$ min, then we have $e^{\mu - \sigma^2} = 10$. \par

(3) In the absence of more information from GTFS, $\sigma$ of the lognormal distribution is sampled uniformly from $0.25 \leq \sigma \leq 0.5$ so that the variance is neither too large or too small. \par

(4) Compute $\mu$ according to the mode and $\sigma$. Shift the distribution rightwards by $t_m$. \par

As $\sigma$ is randomly sampled from a given range, a sensitivity analysis is conducted in Section~\ref{sec: sigma}. If one can get access to the transit vehicle's real-time location, $\sigma$ can be obtained from the real realized travel time data. After computing the travel time distribution between every two adjacent stations, we can obtain the travel time distribution between the origin and all other stations on each transit line by computing the appropriate summation of the random variables due to the independence assumption. \par

The headway at the origin station of each transit line is set to the mean of the difference between the scheduled departure time for every two consecutive trips in our experiment time span. For instance, if the scheduled departure times for a transit line $i$ at the origin station are 8:45, 8:55, 9:10, then we let the headway be $\frac{10 + 15}{2} = 12.5$. The headway for other stations is computed as follows. Let $X_1$ and $X_2$ be the arrival time of the first transit vehicle and the second transit vehicle of the transit line of interest at a station $S$. $O$ is the origin, and $h$ is the deterministic headway at $O$. The cumulative distribution function of the headway is computed as follows: \par

$$P(X_2-X_1\leq t)=P(X_2\leq X_1+t)=\int_0^\infty P(X_2\leq x+t)\cdot P(X_1=x)dx$$
where $P(x_2\leq x+t)$ and $P(X_1 = x)$ are the cumulative distribution function of $X_2$ and the probability density function of $X_1$. Let $t_{O, S}$ be the travel time between $O$ and $S$. The distribution of $X_1$ is the same as the distribution for $t_{O, S}$, and the distribution of $X_2$ is to shift the distribution for $t_{O, S}$ rightwards by $h$. \par

Finally, we estimate the waiting time distribution from the headway distributions following~\cite{larson1981urban}. Let $\mathbb{E}_{i,y}[h_{i,y}]$ be the expected headway of transit line $i$ at station $y$, and $H_{i, y}(\cdot)$ be the cumulative distribution function of headway of transit line $i$ at station $y$. Then the waiting time reads:

$$w_y^i(t,0)=\frac{1}{\mathbb{E}_{i,y}[h_{i,y}]}\cdot (1-H_{i,y}(t))$$

\subsection{Synthetic network experiments}
The synthetic network we use is a 3-line 3-station transit network, as shown in Figure~\ref{fig:small_network}. The three transit lines pass through each of the three stations. The line attributes are presented in Table~\ref{tab: line_attributes}. The headway and travel time are in minutes. \par

\newsavebox{\tempbox}
\newlength{\tempwidth}
\begin{figure*}[htb]
\savebox{\tempbox}{\includegraphics[scale=0.295,clip=true,draft=false,]{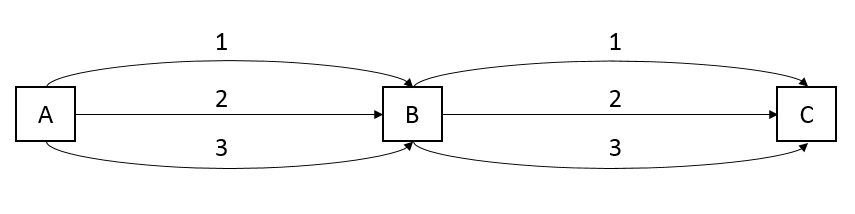}}%
\settowidth{\tempwidth}{\usebox{\tempbox}}%
\hfil\begin{minipage}[b]{\tempwidth}%
\raisebox{-\height}{\usebox{\tempbox}}%
\captionof{figure}{A 3-line synthetic transit network.}%
\label{fig:small_network}%
\end{minipage}%
\savebox{\tempbox}{
\scriptsize{
\begin{tabular}{@{}cccc@{}}
\toprule
Line i & \begin{tabular}[c]{@{}c@{}}Headway\\ at station A\end{tabular} & \begin{tabular}[c]{@{}c@{}}Travel time\\ from A to B\end{tabular} & \begin{tabular}[c]{@{}c@{}}Travel time\\ from B to C\end{tabular} \\ \midrule
1      & 10                                                                  & 4                                                                 & 5                                                                 \\
2      & 15                                                                  & 4                                                                 & 3                                                                 \\
3      & 12                                                                  & 7                                                                 & 4                                                                 \\ \bottomrule
\end{tabular}
}}%
\settowidth{\tempwidth}{\usebox{\tempbox}}%
\hfil\begin{minipage}[b]{\tempwidth}%
\raisebox{-\height}{\usebox{\tempbox}}%
\captionof{table}{Line attributes of the 3-line synthetic transit network. All times are provided in minutes. }%
\label{tab: line_attributes}%
\end{minipage}%
\end{figure*}

\subsubsection{Performance analysis}\label{sec: Ae}
We first illustrate the performance of the algorithms, specifically the relationship between the computation time and the time budget. Let station $A$ be the origin, and station $C$ be the destination, and the time discretization of the algorithm to be $15$ seconds. The three algorithms introduced in Section~\ref{sec: solve} are tested on time budgets ranging from 10 min to 45 min with a step size of 2.5 min. \par

\begin{figure}
\begin{minipage}[t]{0.5\linewidth}
\centering
\includegraphics[width=3.5in]{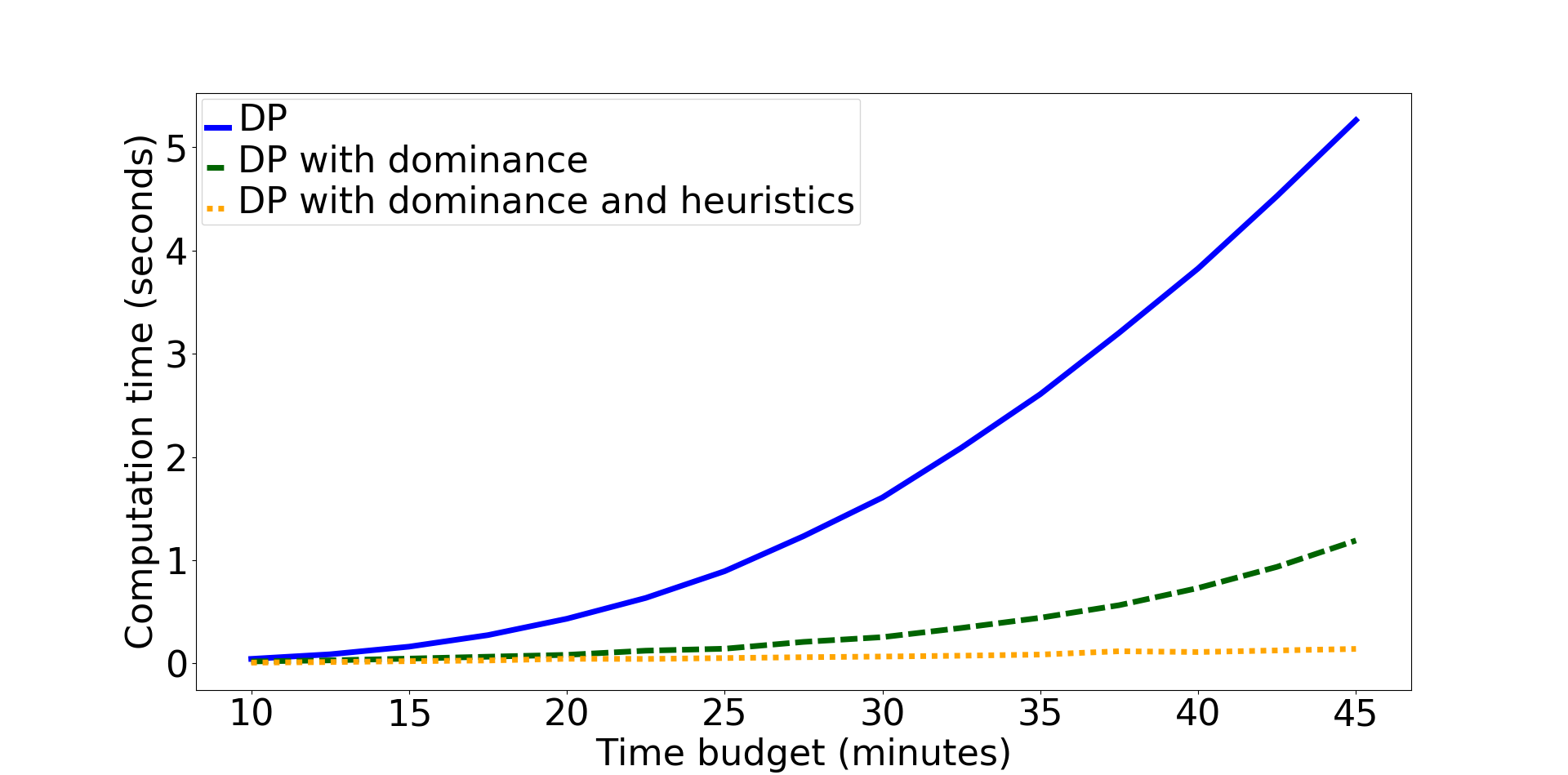}
\end{minipage}%
\begin{minipage}[t]{0.5\linewidth}
\centering
\includegraphics[width=3.5in]{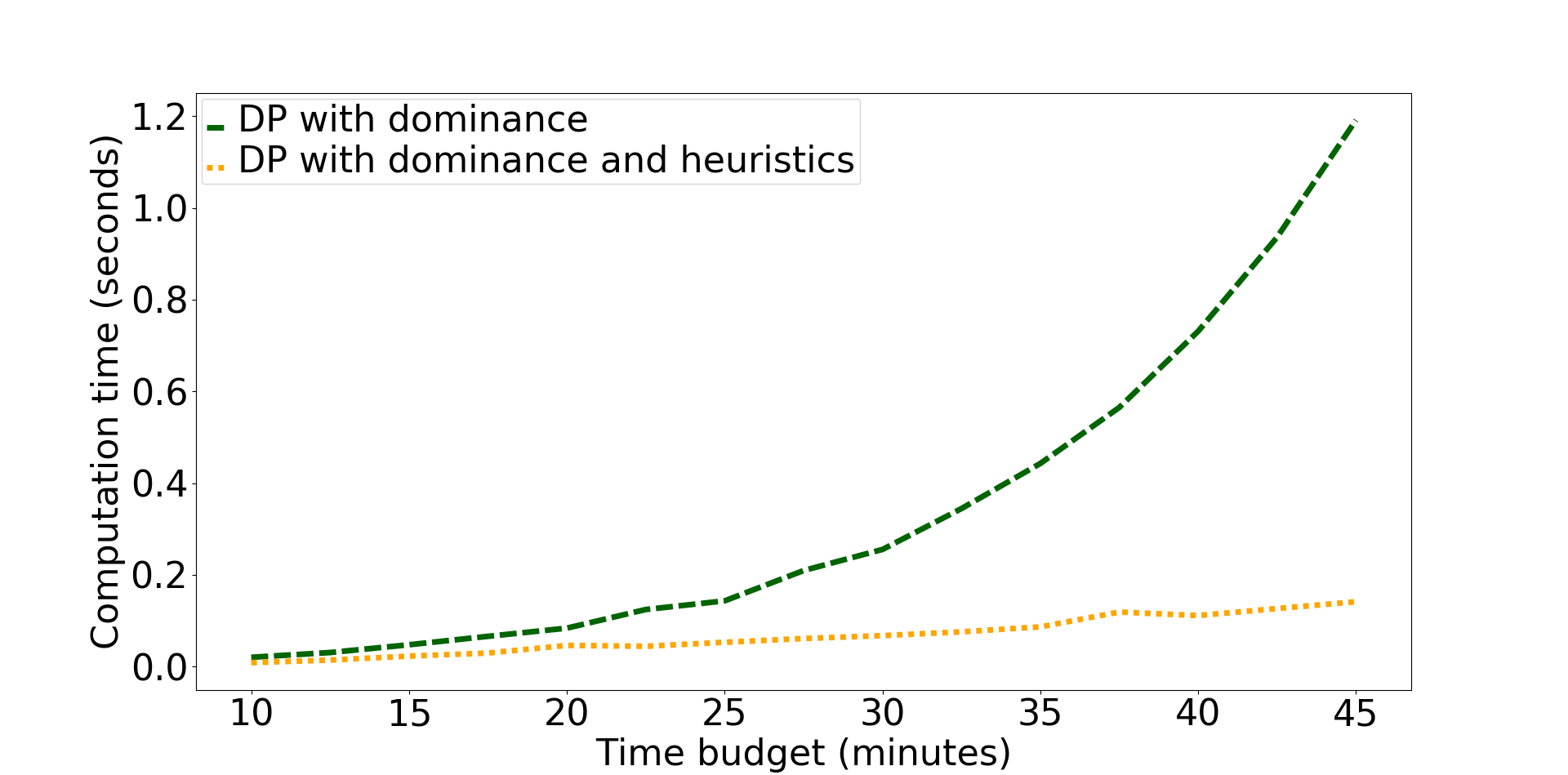}
\end{minipage}
\caption{Computation time of three algorithms on the synthetic network. \textit{Left:} All three algorithms. \textit{Right:} Only algorithms with the dominance based pruning to showcase the relative improvement. }
\label{fig:computation_time}
\end{figure}

As shown in Figure~\ref{fig:computation_time}, the algorithms with search space reduction techniques perform significantly better than the basic DP approach. The DP with the dominance test can provide an average computation time reduction of $77.8\%$ compared to the basic DP approach. The heuristics further decrease the computation time by $68.9\%$ on average relative to the DP with dominance. The time reduction percentage ($\frac{\text{Computation time of DP - Computation of DP with dominance}}{\text{Computation time of DP}}$) is typically higher when the time budget is large. When the time budget is 45 min, the time reduction can be $97.3\%$ for the DP with dominance and heuristics. The average relative error for the heuristic method is only $2.8\%$ in this experiment. Therefore, the heuristic (at least in this case) provides a satisfactory trade-off between accuracy and computation time. \par

\begin{figure}
\centering
\includegraphics[width=3.6in]{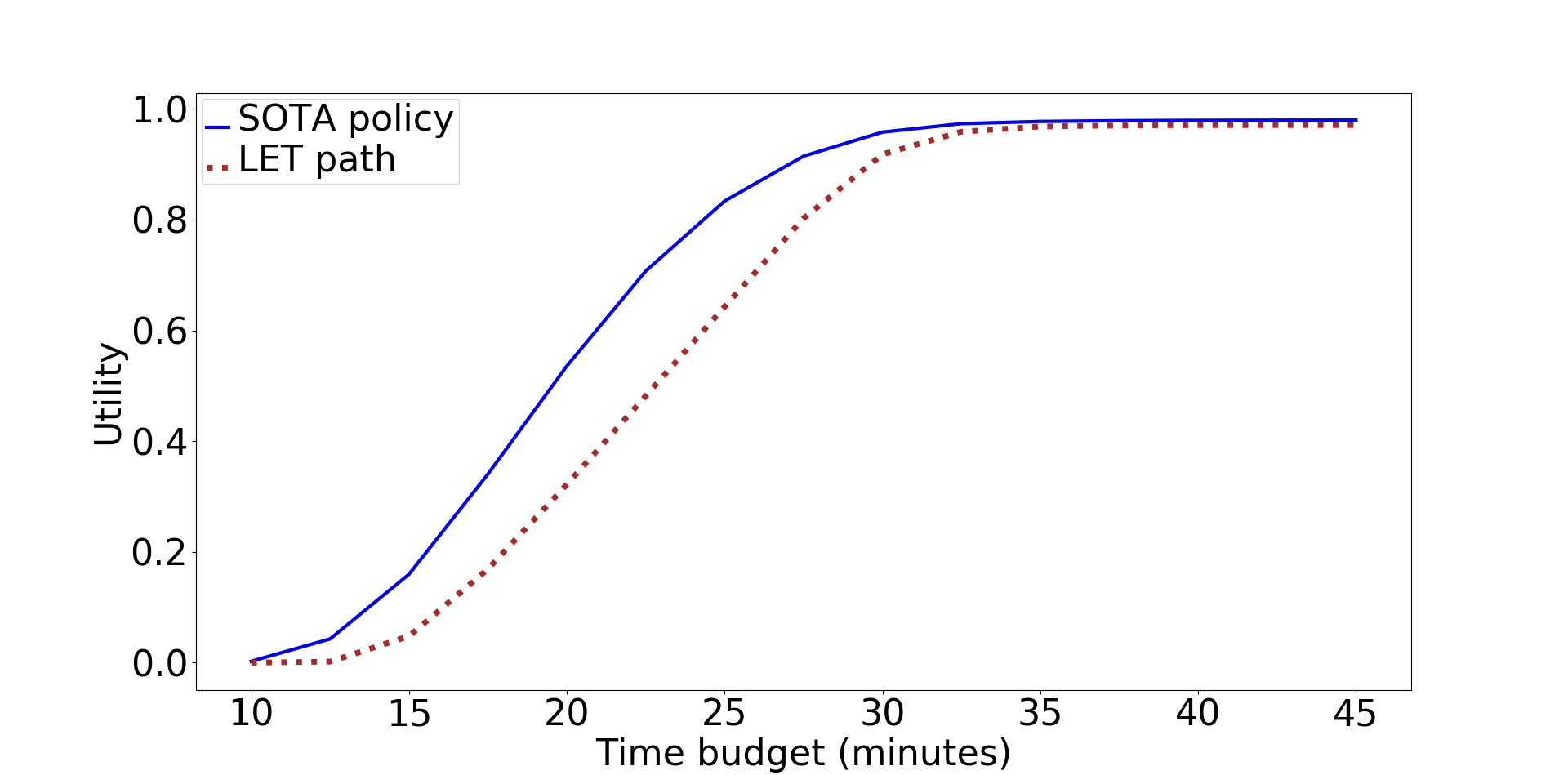}
\caption{Comparison of the utility between the SOTA policy and LET path.}
\label{fig:u_diff_small}
\end{figure}

We also quantify the benefits of utilizing a SOTA policy instead of using the least expected travel time (LET) path. This is done by using the expected value of both the travel time and waiting time distributions to compute the LET path between station $A$ and station $C$. The utility of the LET solution is then computed based on the probability of success when using the LET path for different time budgets. Figure~\ref{fig:u_diff_small} illustrates the utility difference between the LET path and the SOTA policy we obtain from our algorithm. When the time budget is very small, the utility difference is zero because the passenger cannot reach the destination on time regardless of the policy the passenger uses. Similarly, when the time budget is very large, the utility difference is also 0 because the passenger reaches the destination on time with high probability even with a sub-optimal route. The SOTA policy will always be no worse than the LET solution by definition. The maximum utility difference observed in this experiment is 23 percent (when the time budget is 22.5 min in this case). \par

\subsubsection{Sensitivity to the modes of the travel time distributions. }\label{sec: modes}
In this section, we analyze the effect of the modes\footnote{Note that \textit{mode} here is the statistical definition.} of travel time distributions on computation time. The modes of the travel time distributions on each link are sampled uniformly from a range instead of being predefined. Two sets of experiments are conducted. In the first set of experiments, we let the width of the range be 1. More specifically, we first generate a random number $i \in \{1,2,3,4,5,6,7,8,9\}$ uniformly, and then set the range to be $[i,i+1)$. Assume that $X_1$ and $X_2$ are two random variables sampled uniformly from this range. Then, $E(|X_1-X_2|) = \frac{1}{3}$ since $|X_1-X_2|$ follows a triangle distribution with the parameters $a=0$, $c=0$ and $b=1$. In the second set of experiments, the range is set to $[1,10)$, and $E(|X_1-X_2|)=3$ since $|X_1-X_2|$ follows a triangle distribution where $a=0$, $c=0$ and $b=9$. We call the first set of experiments \emph{low diff}, and the second set of experiments \emph{high diff}. Each set of experiments are conducted for 100 times. \par

\begin{figure}[htb!]
\centering
\includegraphics[width=0.65\textwidth]{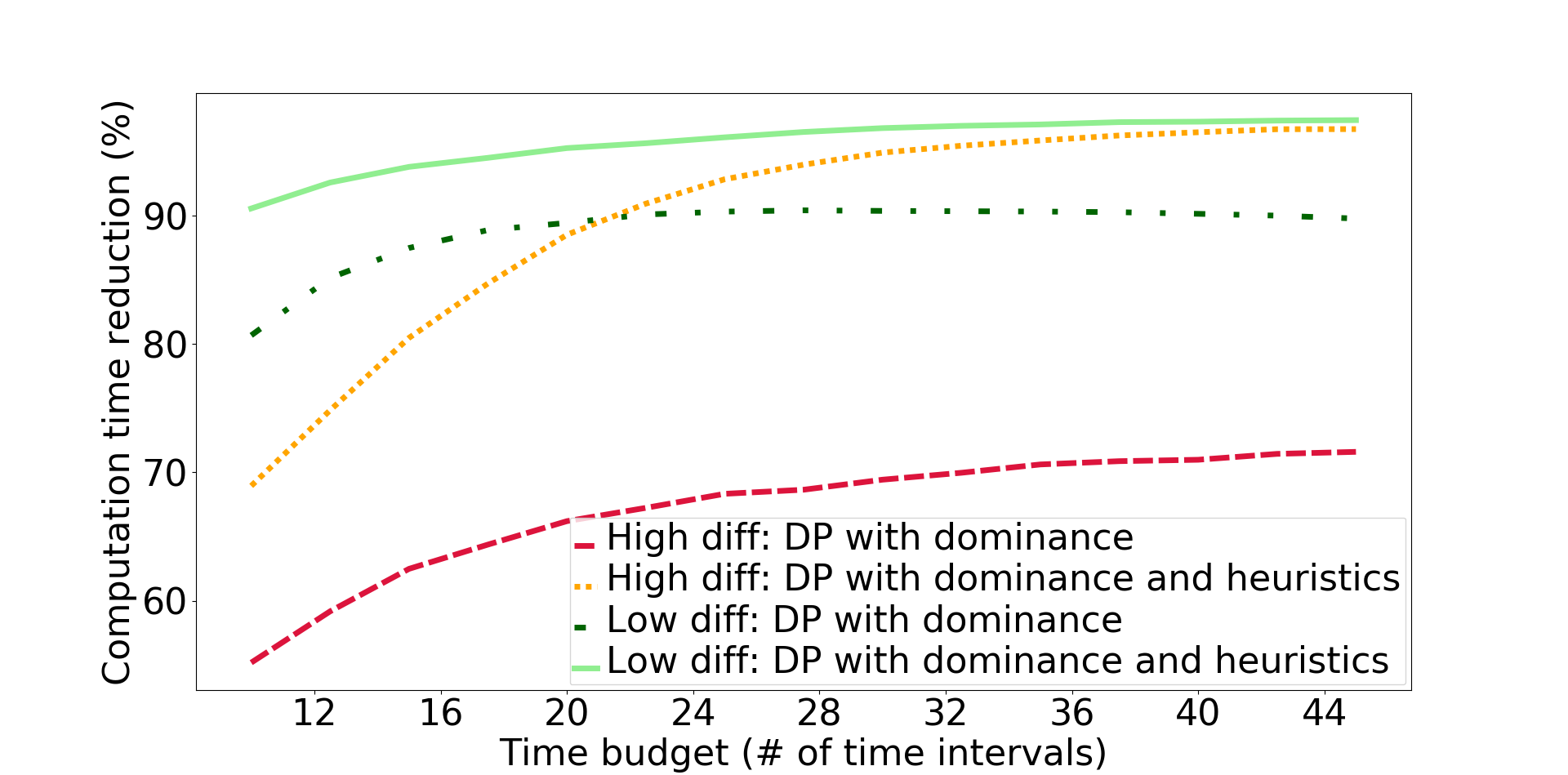}
\caption{\label{fig:highlow} Computation time reduction of speed-up techniques on two cases. \emph{high diff} represent the case where the travel time distributions' modes are highly different, while \emph{low diff} represents the case where the travel time distributions' modes are relatively close. }
\end{figure}

The computation time reduction with respect to the time budget is shown in Figure~\ref{fig:highlow}. The search space reduction techniques work well in both sets of experiments. The average computation time reductions for the DP with dominance are $88.9\%$ and $67.1\%$ for the \emph{low diff} and \emph{high diff} cases respectively relative to standard DP, and the time reductions are $95.7\%$ and $89.8\%$ for the DP with dominance and heuristics. The search space reduction techniques work better in the \emph{low diff} cases because the transit services that arrive the first will intuitively dominate the other transit lines more often with similar travel time distributions, since waiting is not likely to increase the utility. Another interesting result is that the heuristics will provide a higher time reduction for the \emph{high diff} cases compared to the DP with dominance algorithm. A potential reason is that a line will have a higher probability of being better than any other single candidate line at the same station in this case, which indicates dominance according to the Heuristic~\ref{h: anyother}. \par


\subsubsection{Sensitivity to the $\sigma$ parameter} \label{sec: sigma}
We now analyze the sensitivity of the algorithm performance to the parameter $\sigma$ in the lognormal travel time distributions. Assume that the CDF of a lognormal distribution is $F(x)$, and $\gamma$ is the mode of the distribution. In Section~\ref{sec: Ae}, we use $\sigma=0.25$, which makes $F(1.5\cdot \gamma)\approx 0.9$ when $\gamma$ is in the range shown in Table \ref{tab: line_attributes}. To show the effects of $\sigma$ on the distribution, we compare the distribution with $\sigma = 0.25$ and $\sigma = 0.5$ in Figure~\ref{fig:lognormal}. When $\sigma=0.5$, $F(1.5\cdot \gamma)\approx 0.6$. \par

\begin{figure}[htb]
  \centering
  \includegraphics[width=0.6\linewidth]{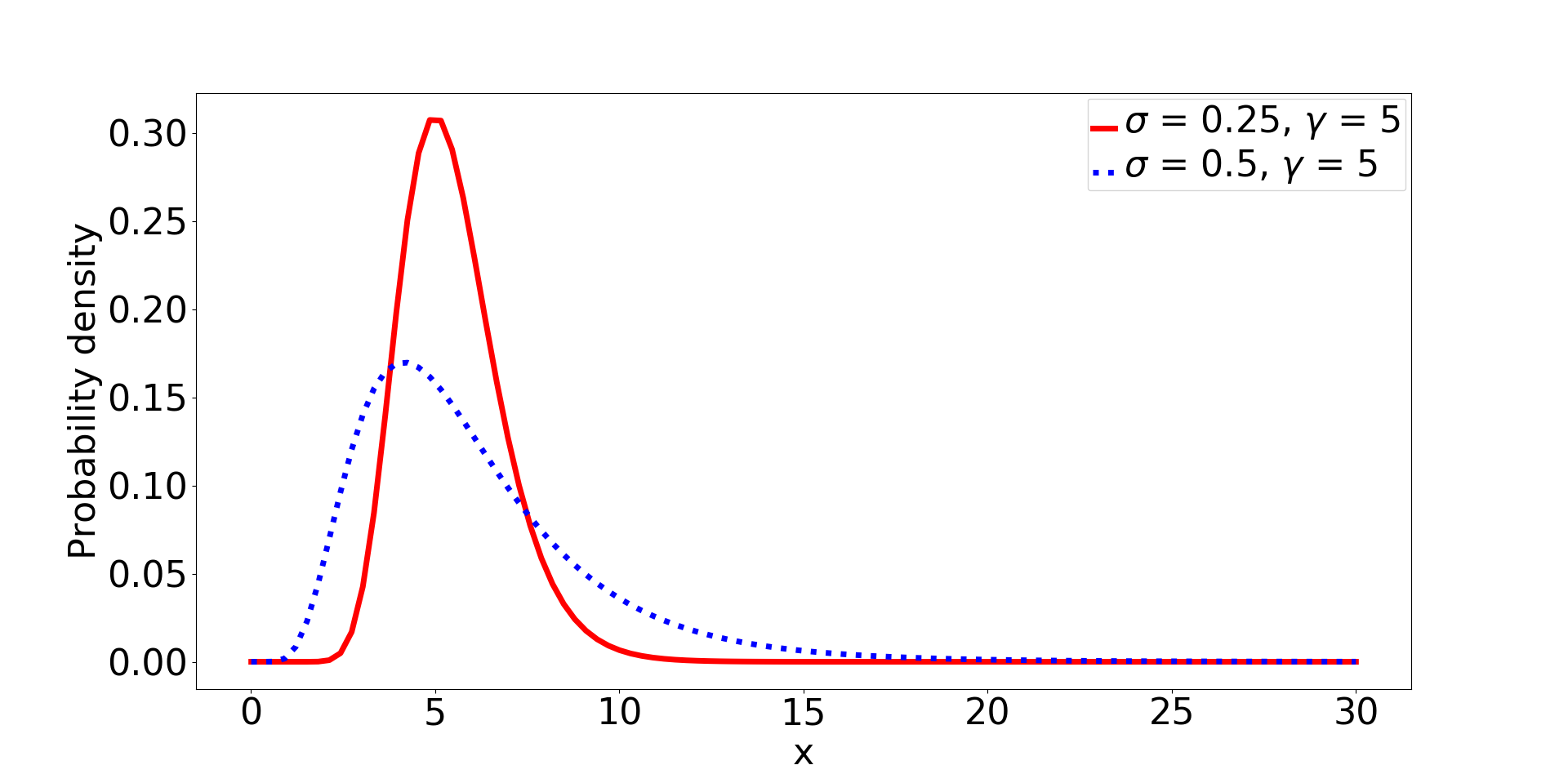}
  \captionof{figure}{The PDF of lognormal distributions with different $\sigma$.}
  \label{fig:lognormal}
\end{figure}

To analyze the sensitivity of $\sigma$ on the algorithm's performance, two sets of experiments are conducted. In the first set of experiments, the algorithms are tested under $\sigma=0.5$. In the second set of experiments, $\sigma$ is uniformly sampled from $[0.25,0.5]$ for each link of each transit line. Again, each set of experiments are conducted for 100 times. \par

\begin{table}[htb]
\centering
\caption{Algorithms' performance under different $\sigma$ settings.}
\label{diffsigma}
\begin{tabular}{@{}lll@{}}
\toprule
Parameter setting                                                                    & \begin{tabular}[c]{@{}l@{}}Computation time reduction\\  for DP with dominance\end{tabular} & \begin{tabular}[c]{@{}l@{}}Computation time reduction\\ for DP with dominance and heuristics\end{tabular} \\ \midrule
$\sigma=0.25$                                                                        & $77.8\%$                                                                                    & $92.6\%$                                                                                           \\
$\sigma=0.5$                                                                         & $80.3\%$                                                                                    & $94.1\%$                                                                                           \\
\begin{tabular}[c]{@{}l@{}}$\sigma$ randomly\\  chosen from $[0.25,0.5]$\end{tabular} & $79.8\%$                                                                                    & $93.9\%$                                                                                           \\ \bottomrule
\end{tabular}
\end{table}

As shown in Table~\ref{diffsigma}, the search space reduction techniques perform similarly under different $\sigma$ values. The computation time is reduced slightly more when the variance of the travel time distribution is larger. In the following experiments for the Chicago transit network, we use $\sigma$ randomly sampled from $[0.25, 0.5]$. \par

\subsection{Chicago network experiments}
We now use the Central/South region (including downtown) of the Chicago transit network as a case study. The network is created according to the GTFS data published by the Chicago Transit Authority (CTA) through Google’s GTFS project. The GTFS data contains schedules and associated geographic information. In the experiments, only active bus lines in the morning (6 am to 10 am) are considered. The transit network, which contains 49 transit lines and 1565 stations, is shown in Figure~\ref{fig:chicago}. \par

\begin{figure}[htb]
\centering
\includegraphics[width=0.4\textwidth]{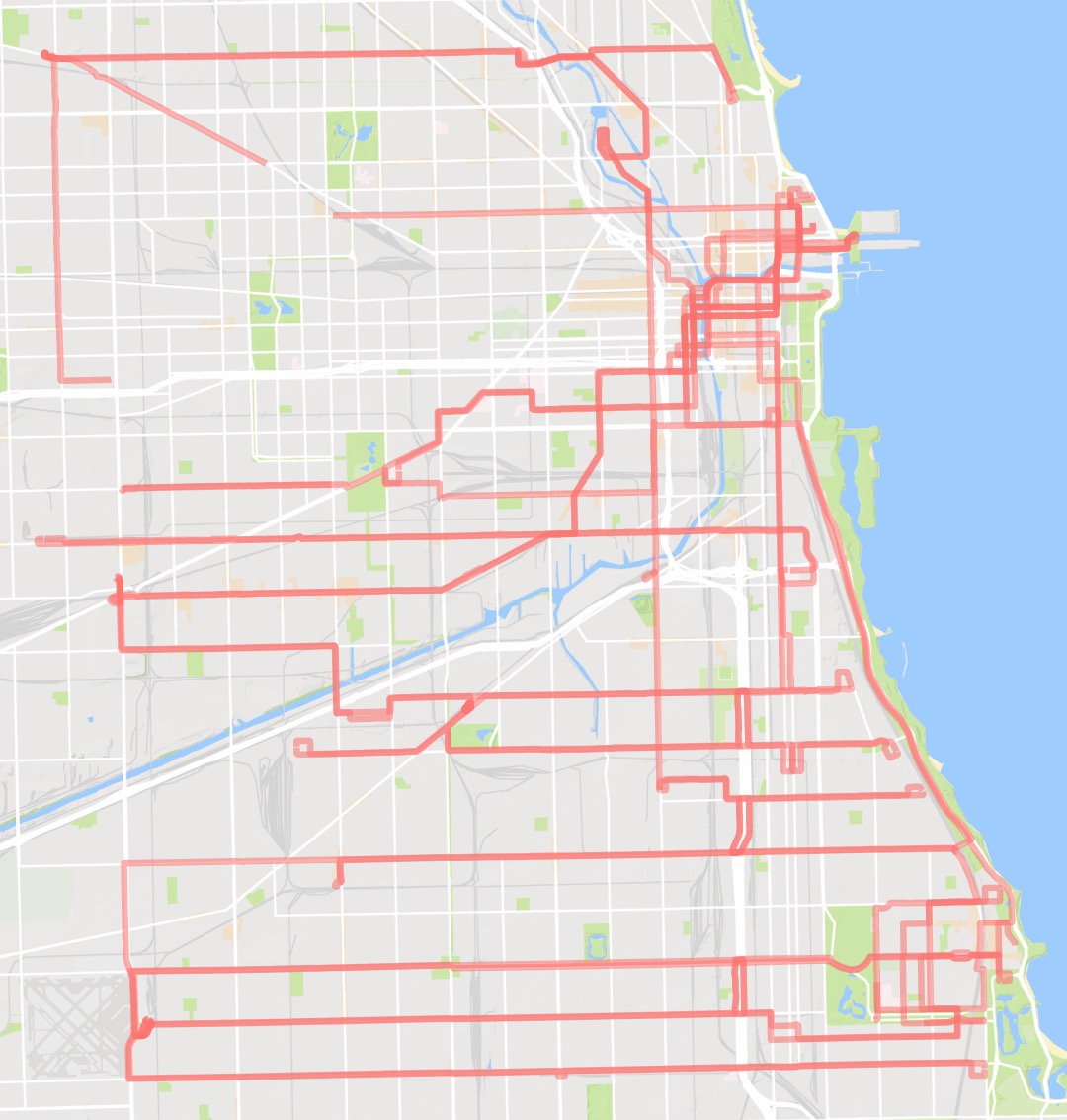}
\caption{\label{fig:chicago} Chicago transit network (Central/South region).}
\end{figure}

\subsubsection{Runtime performance in the Chicago network} \label{sec: Ae-chi}
The average travel time to commute in the United States is 26.1 min according to the U.S. Census Bureau \citep{www.census.gov}. Therefore, we select 100 ODs randomly from the station set under the constraints: (1) The expected trip duration is from 15 min to 45 min; (2) At least one of the origin and the destination is in the downtown area, to simulate the case of commuting. The experiments are conducted for the time budgets which also span from 10 min to 45 min. For each time budget, we use the same three algorithms described previously to compute the utility. The computation time for each algorithm is shown in Figure~\ref{fig:chicago_efficiency}. \par

\begin{figure}
\begin{minipage}[t]{0.5\linewidth}
\centering
\includegraphics[width=3.5in]{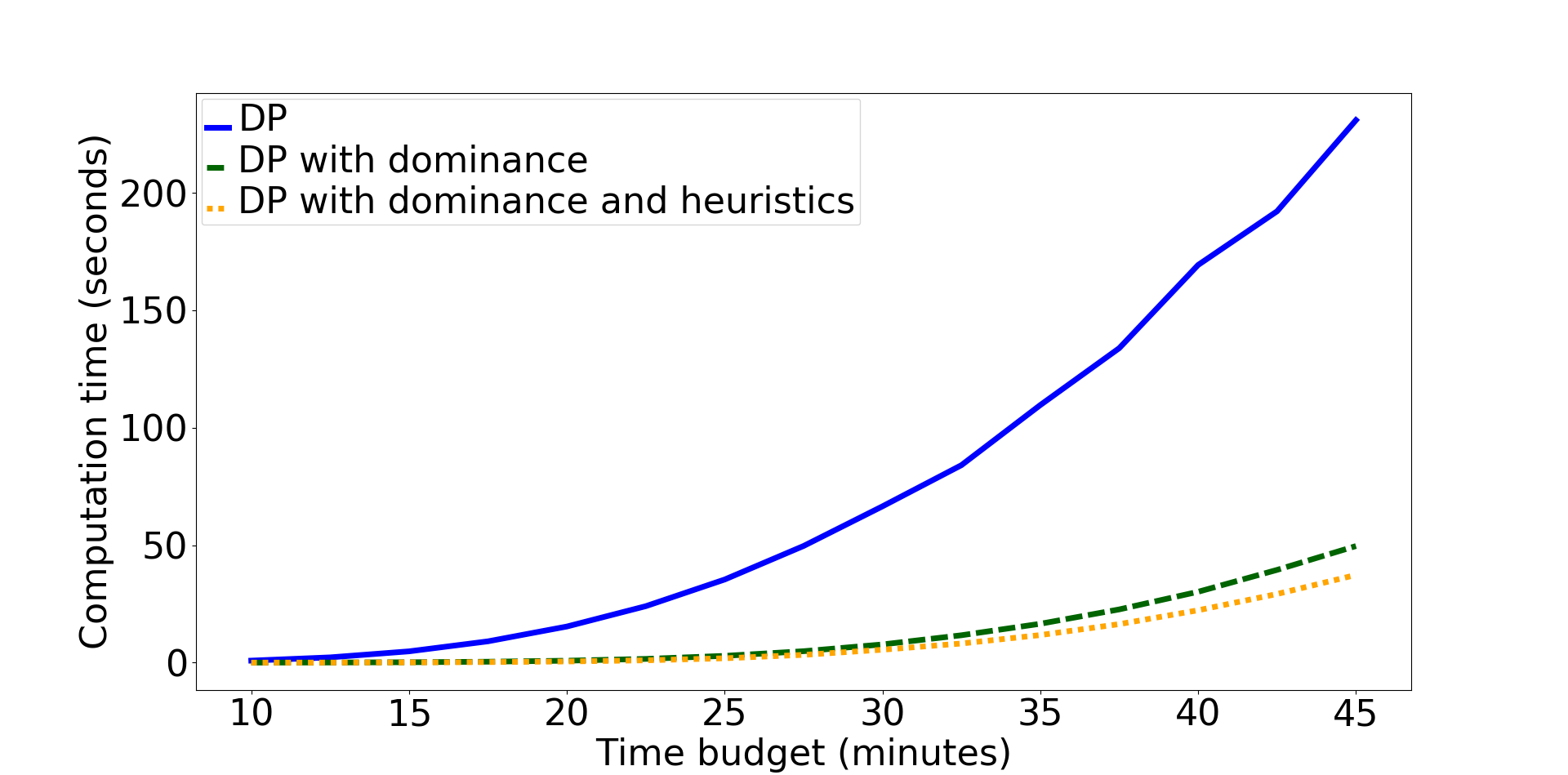}
\end{minipage}%
\begin{minipage}[t]{0.5\linewidth}
\centering
\includegraphics[width=3.5in]{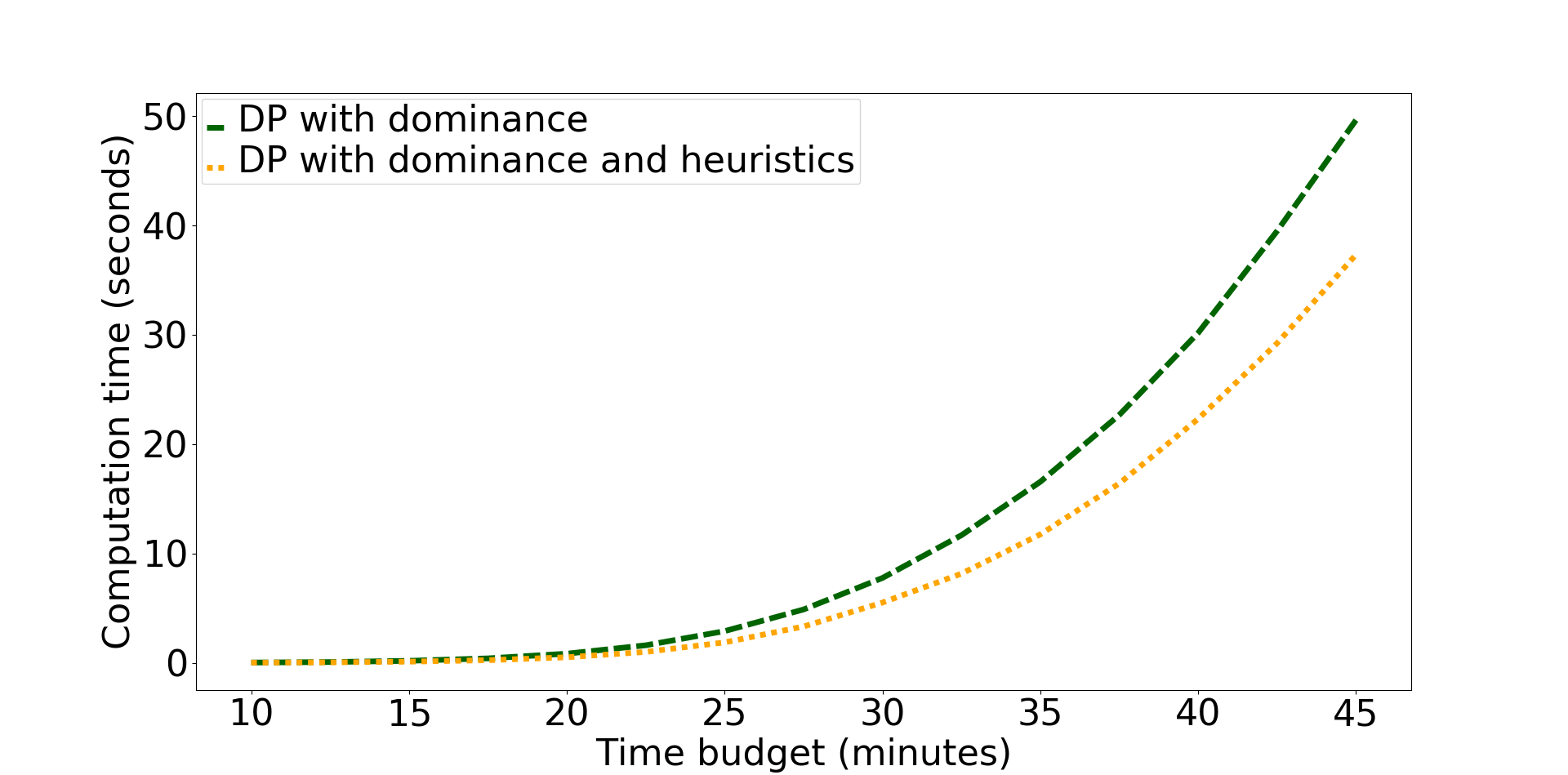}
\end{minipage}
\caption{Computation time of three algorithms on the Chicago transit network. \textit{Left:} All three algorithms. \textit{Right:} Only algorithms with the dominance based pruning to showcase the relative improvement. }
\label{fig:chicago_efficiency}
\end{figure}

The DP with dominance reduces the computation time by $89.3\%$ on average. The computation time reductions are significant, but not as high as in the synthetic network shown in Figure~\ref{fig:small_network} especially when $T$ is large. In the synthetic network, there are only three stations which limit the search space. However, in real-size networks, the search space will be much larger since there are more stations and transit lines to be considered. The heuristic rules further decrease the computation time by $34.0\%$ on average relative to the DP with dominance. The average relative error for the heuristic rules is about $1.6\%$. \par

\begin{figure}[htb]
\centering
\includegraphics[width=0.65\textwidth]{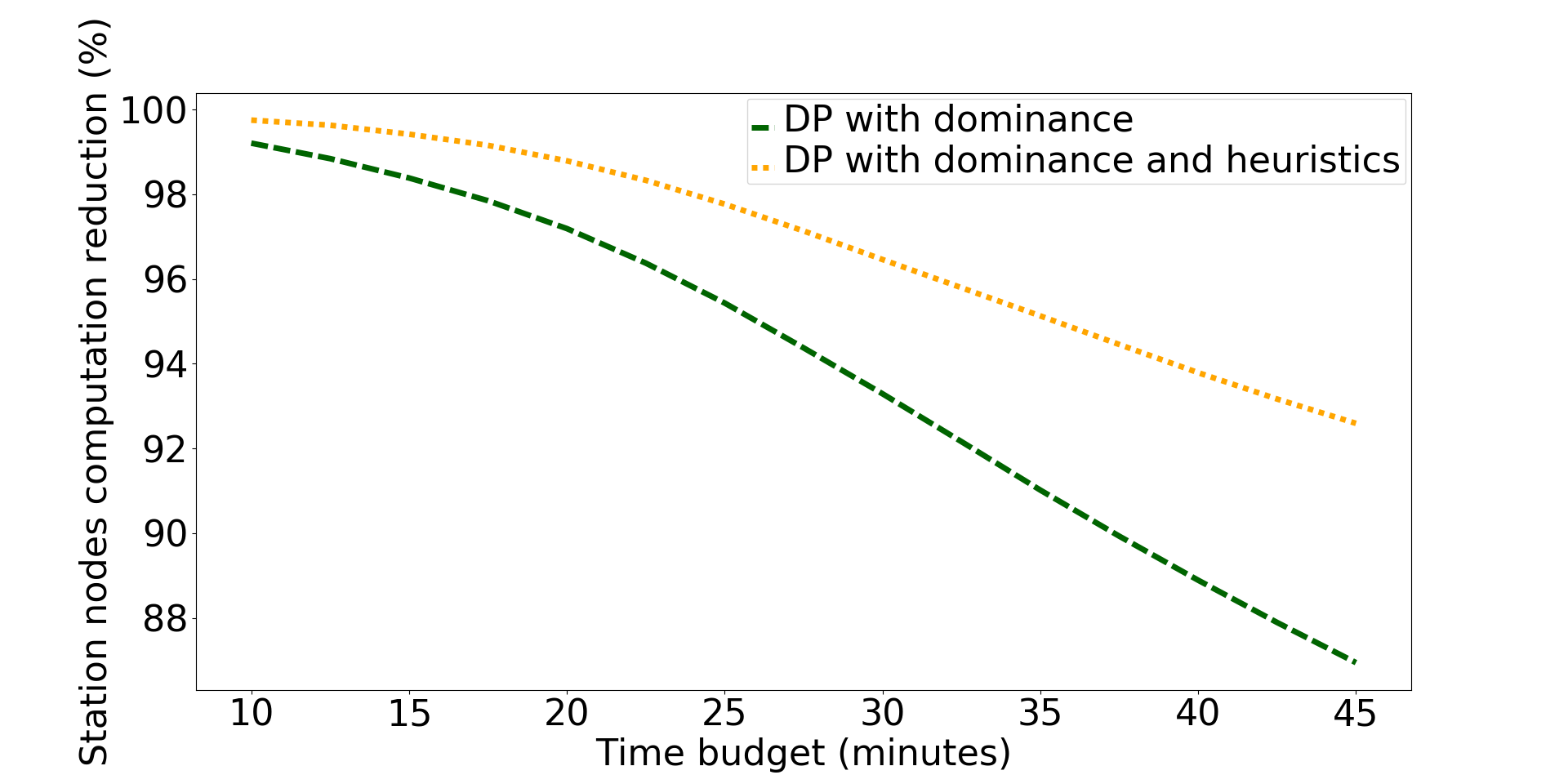}
\caption{\label{fig:station_reduction} The computation reduction for computing the utility on station nodes. The vertical axis shows the percentage of function calls reduction to compute the utility at stations. }
\end{figure}

Recall from Section~\ref{sec: space}, if $i\succeq X$ at an arrival node $A_y^{i, X}$ given passenger state $(t, r)$, we save the computation for $u_{S_y^{X}}(t, r)$. In the experiments, we also record the number of function calls to compute the utility at station nodes. Figure~\ref{fig:station_reduction} shows the percentage of the function calls reduction by DP with dominance and DP with dominance and heuristics. Both algorithms provide more than $90\%$ function calls reduction for all the cases. However, the reduction decreases as the time budget increases. When the time budget is higher, the set of feasible arrival nodes that followed by a station node $S_y^{X}$ at passenger state $(t, r)$ is larger. Here an arrival node $A_y^{i, X}$ is \textit{feasible} at passenger state $(t, r)$ if it is not pruned by the techniques in Section~\ref{sec: space}. Let $Z_y^X$ be the set of feasible arrival nodes followed by station node $S_y^{X}$. To save the computation for $u_{S_y^{X}} (t, r)$, every line in $\{i \mid A_y^{i, X} \in Z_y^{X}\}$ needs to dominate $X$, which is of less probability when $|Z|$ is higher. \par

The average computation time of the DP with dominance and heuristics is only about 9.2 seconds on a personal computer, which indicates the potential of running the algorithm in real-time applications. \par

\subsubsection{Compare the utility of the SOTA policy and LET path}
\begin{figure}[htb]
\centering
\includegraphics[width=0.65\textwidth]{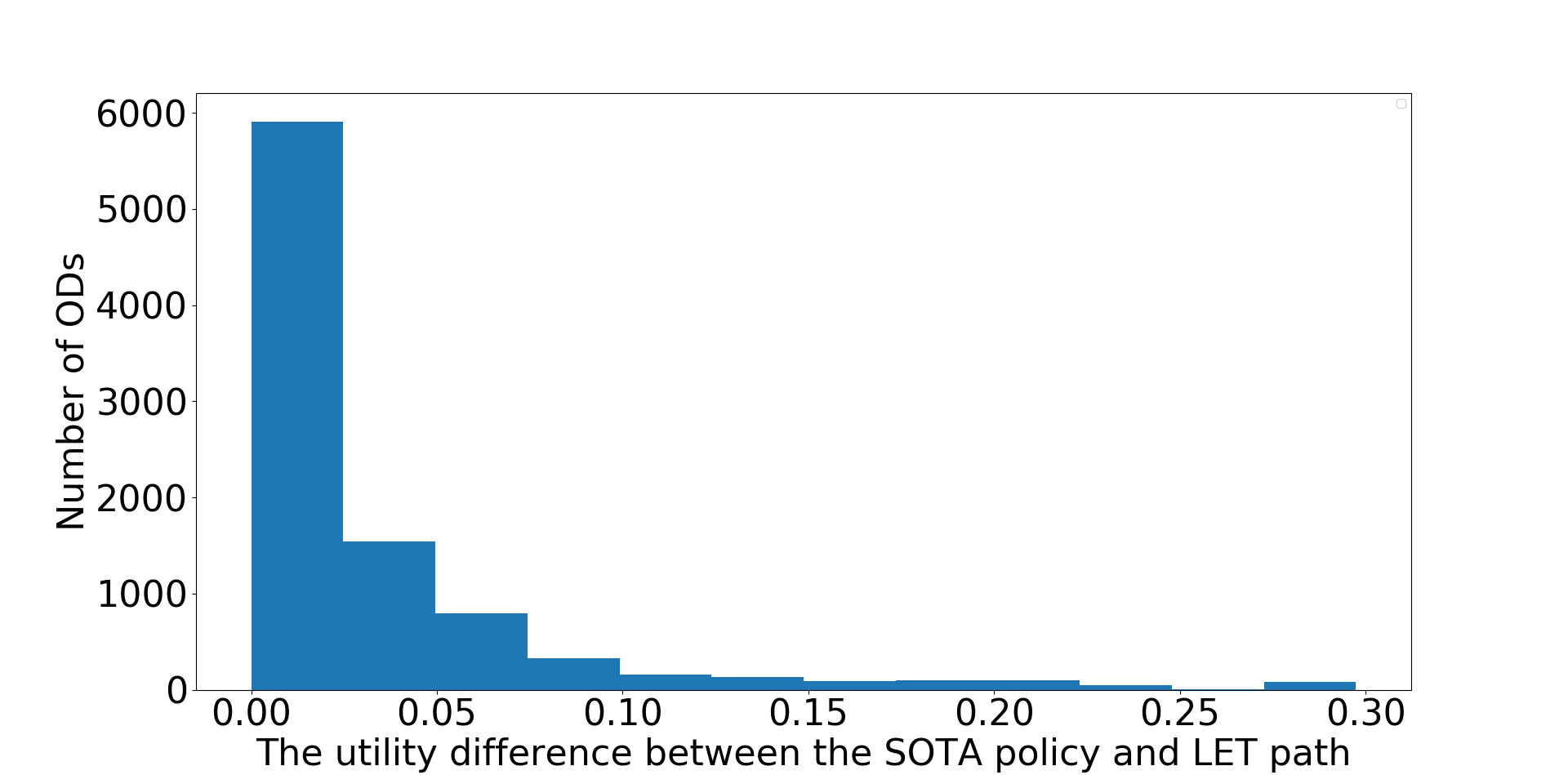}
\caption{\label{fig:utility_diff} The histogram of the utility difference between the SOTA policy and LET path on the Chicago transit network.}
\end{figure}

Recall from Section~\ref{sec: Ae}, the utility difference between the SOTA policy and LET path can be as large as 0.23 in the synthetic network. In this section, we test the utility difference in the Chicago transit network. Intuitively, if there is one and only one transit line that can directly take the passenger to the destination without transferring, the SOTA policy and LET path will likely be the same due to the time cost of transferring. Therefore, we compare the utility difference of all the ODs from the station set under the constraints: (1) The expected trip duration is from 15 min to 45 min; (2) At least one of the origin and the destination is in the downtown area; (3) There is no transit line or there are more than one transit line that can directly take the passenger to the destination. The number of ODs satisfying the above conditions is 9290. Again, the experiments are conducted for the time budgets from 10 min to 45 min. For each OD, we record the maximum utility difference between the SOTA policy and LET path among all the time budgets we test. Figure~\ref{fig:utility_diff} shows the histogram of the maximum utility difference. The utility difference for $19.6\%$ of the ODs is larger than 0.05, and the utility difference for $7.6\%$ of the ODs is larger than 0.1. The utility difference is close to 0 for most of the ODs, which are the cases that either there is only one route choice between the origin and the destination or there is one obvious better line than the other choices. It should be noted that the results highly depend on the network and the travel time distribution we use. \par

\begin{figure}
\begin{minipage}[t]{0.5\linewidth}
\centering
\includegraphics[width=3.6in]{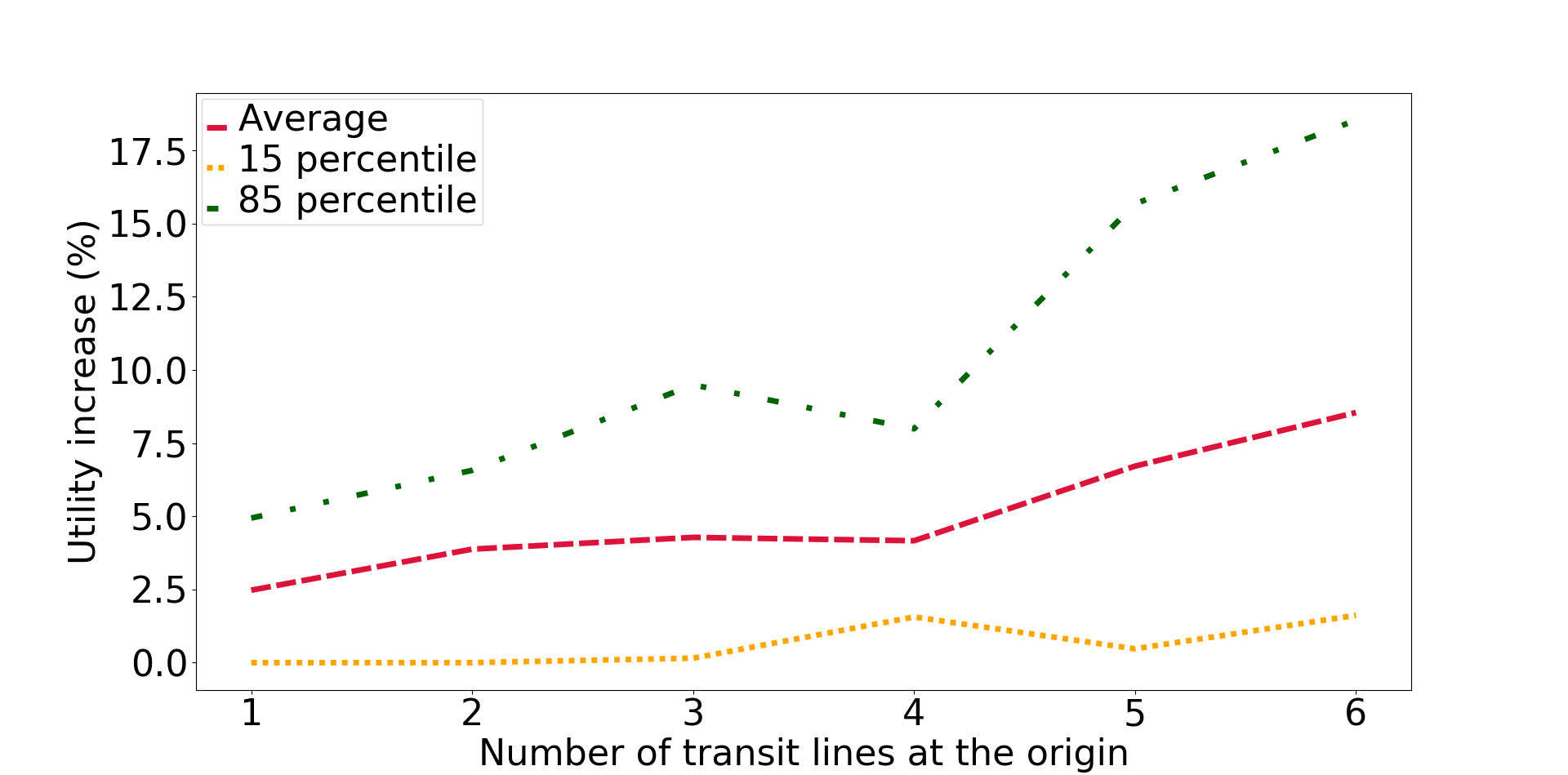}
\end{minipage}%
\begin{minipage}[t]{0.5\linewidth}
\centering
\includegraphics[width=3.6in]{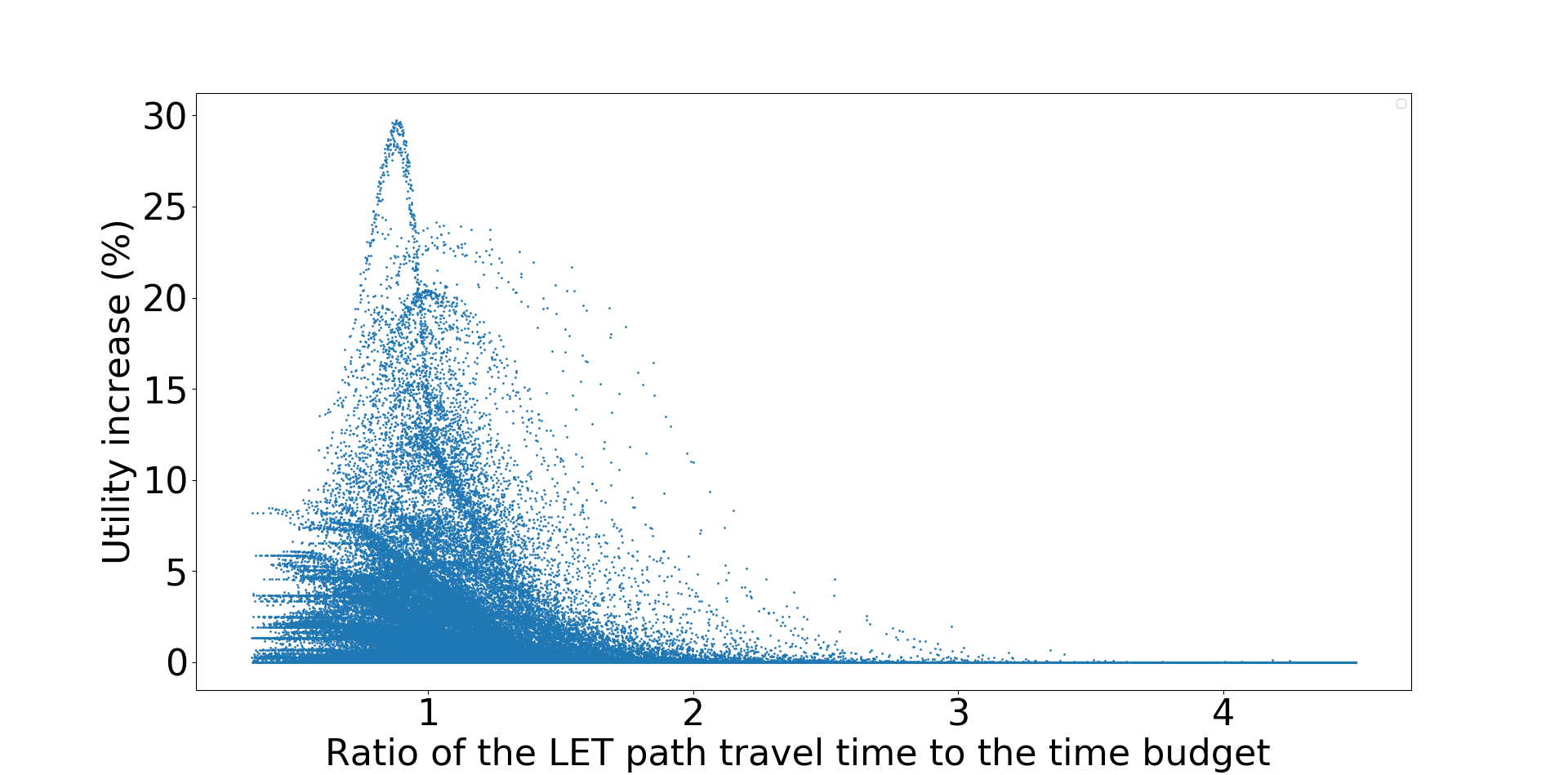}
\end{minipage}
\caption{The utility improvement of the SOTA policy over the LET path under different circumstances. \textit{Left:} With respect to the number of transit lines passing by the origin. \textit{Right:} With respect to the ratio of the travel time of the LET path and the time budget. }
\label{fig:n_start_stations}
\end{figure}

To give insights of what trips will likely benefit from the SOTA policy, we show the utility increases of the SOTA policy over the LET path on different circumstances with respect to: i) The number of transit lines passing by the origin; ii) The ratio of the least expected time (LET) path travel time to the time budget. As shown in Figure~\ref{fig:n_start_stations} (Left), a SOTA policy outperforms the LET path more when the number of transit lines at the origin increases, which demonstrates that a SOTA policy is more advantageous when there are more states and options to consider. In addition, even though there might be more transit lines to transfer on the way, it may not be beneficial to conduct a transfer in most of the cases due to the transfer time. Therefore, when there is only one transit line passing by the origin, the average utility difference is about 0.025, which is close to 0. The ratio of the LET path travel time to the time budget represents how sufficient the time budget is relative to the trip length. Figure~\ref{fig:n_start_stations} (Right) shows the utility increase versus the ratio. If the ratio is relatively large (larger than 2), passengers cannot increase the utility by using the adaptive routing policy since it is difficult to get to the destination in a very limited time budget; On the other hand, if the ratio is very small (smaller than 0.5), the utility increase is also close to 0 since passengers will likely reach the destination within the given time budget whatever transit line they board. The SOTA policy performs the best when the ratio is around 1. In summary, the SOTA policy in transit networks shines the most when there are more states to consider and when the time budget is neither too high or too low relative to the trip length. \par

\section{Conclusion.}\label{sec:conclusion}
In this article, we extend the SOTA formulation to the case of transit networks. A network representation comprised of three types of nodes is proposed to model the common line problem in the SOTA framework. Instead of assuming that the passenger will board the first arriving transit service in an attractive line set that is precomputed, we give a routing policy which outperforms a-priori solutions for all practical purposes. From the perspective of computation, we show that computing the utility in transit networks is significantly more difficult than in road networks, and design an dynamic programming based algorithm to solve the problem, which is pseudo-polynomial in the number of stations and time budget, and exponential in the number of transit lines at a station, which is practically a small number. \par

To reduce search space, we propose a definition of transit line dominance, and a series of techniques to check if a transit line dominates. Experiments are conducted on both a synthetic network and the Chicago transit network. The results show that the search space reduction techniques can reduce the computation time by up to around $90\%$, and the heuristic rules can further decrease the computation time by about $34\%$. Two sensitivity analyses on the travel time distribution show that the search space reduction techniques can significantly decrease the computation time in all the cases we test. Finally, the algorithms are applied in the Chicago transit network, and similar conclusions are obtained. \par

In future research, we will focus on designing proper preprocessing techniques to further reduce the computation time in practice. In reality, the spatial and temporal correlation of link travel time usually exists in transit networks. In addition, the travel time and waiting time can be time-varying at the different time of a day (e.g., at peak hours or off-peak hours). Therefore, one possible future research direction is to explore the ways to generate more realistic travel time and waiting time distributions and design algorithms to deal with them. \par



 \bibliographystyle{elsarticle-harv}

\bibliography{sample}

\end{document}